\documentclass[a4paper]{amsart}

\usepackage{amsmath,amssymb,microtype}
\usepackage[nameinlink,capitalize]{cleveref}
\usepackage{enumerate}
\usepackage{url}
\usepackage{tikz}
\usetikzlibrary{automata,positioning}

\newcommand{\A}{\mathcal{A}}
\newcommand{\B}{\mathcal{B}}
\newcommand{\D}{\mathsf{D}}
\newcommand{\R}{\mathsf{R}}
\newcommand{\N}{\mathbb{N}}
\newcommand{\F}{\mathsf{F}}
\newcommand{\V}{\mathsf{V}}
\newcommand{\Freg}{\mathsf{F_{\mathsf{reg}}}}
\newcommand{\Vreg}{\mathsf{V_{\mathsf{reg}}}}
\newcommand{\last}{\mathrm{last}}
\newcommand{\wnd}{\mathrm{wnd}}

\newcommand{\rev}{^\mathsf{R}}
\renewcommand{\det}{^\mathsf{D}}
\newcommand{\revdet}{^\mathsf{RD}}

\newcommand{\NL}{\textsc{NL}}
\newcommand{\PSPACE}{\textsc{Pspace}}

\theoremstyle{plain}
\newtheorem{thm}{Theorem}[section]
\newtheorem{defn}[thm]{Definition}
\newtheorem{prop}[thm]{Proposition}
\newtheorem{lem}[thm]{Lemma}
\newtheorem{cor}[thm]{Corollary}

\Crefname{thm}{Theorem}{Theorems}
\Crefname{defn}{Definition}{Definitions}
\Crefname{lem}{Lemma}{Lemmas}
\Crefname{prop}{Proposition}{Propositions}
\Crefname{cor}{Corollary}{Corollaries}
\Crefname{ex}{Example}{Examples}

\title{Automata theory on sliding windows}

\author[M.~Ganardi]{Moses Ganardi}
\author[D.~Hucke]{Danny Hucke}
\author[D.~K\"onig]{Daniel K\"onig}
\author[M.~Lohrey]{Markus Lohrey}
\address{Universit\"at Siegen, Germany\\
  \texttt{\{ganardi,hucke,koenig,lohrey\}@eti.uni-siegen.de}}
\author[K.~Mamouras]{Konstantinos Mamouras}
\address{University of Pennsylvania, Philadelphia, USA\\
  \texttt{mamouras@seas.upenn.edu}}

\begin{document}

\begin{abstract}
	In a recent paper we analyzed the space complexity of streaming algorithms whose goal is to decide
	membership of a sliding window to a fixed language. For the class of regular languages 
	we proved a space trichotomy theorem: for every regular language the optimal space bound is
	either constant, logarithmic or linear. 
	In this paper we continue this line of research: We present natural characterizations for the constant
	and logarithmic space classes and establish tight relationships to the concept of language growth. We also
	analyze the space complexity with respect to automata size and prove almost matching lower and upper bounds.
	Finally, we consider the decision problem whether a language given by a DFA/NFA admits a sliding window algorithm
	using logarithmic/constant space.
\end{abstract}

\maketitle

\section{Introduction}

\subsection{Streaming algorithms}
Streaming algorithms process an input sequence $a_1 a_2 \cdots a_m$ from left to right 
and  have at time $t$ only direct access to the current data value $a_t$.
Such algorithms have received a lot of attention in recent years, see~\cite{Aggarwal07} for a broad introduction.
The general goal of streaming algorithms is to avoid the explicit storage of the whole data stream.
Ideally, a streaming algorithm works in constant space, in which case it reduces to a deterministic finite
automaton (DFA), but polylogarithmic space with respect to the input length might be acceptable, too.
These small space requirements are motivated by the current explosion in the size of the input data,
which makes random access to the input often infeasible.
Such a scenario arises for instance when searching in large databases (e.g., genome databases or web databases),
analyzing internet traffic (e.g. click stream analysis), and monitoring networks. 
	
The first papers on streaming algorithms as we know them today are usually attributed to Munro and Paterson \cite{MunroP80}
and Flajolet and Martin \cite{FlajoletM85}, although the principle idea goes back to the work on online machines
by Hartmanis, Lewis and Stearns from the 1960's \cite{LewisSH65,StearnsHL65}.
Extremely influential for the area of streaming algorithms was the paper of Alon, Matias, and Szegedy \cite{AlonMS99}.

\subsection{The standard model and sliding window model.}
Two variants of streaming algorithms can be found in the literature:
\begin{itemize}
	\item In the {\em standard model} the algorithm reads an input stream $a_1 a_2  \cdots a_m$
	of data values from left to right.
	At time instant $t$ it has to output the value $f(a_1 a_2 \cdots a_t)$ for a certain function $f$.
	\item In the {\em sliding window model} the algorithm works on a sliding window.
	At time instant $t$, the active window is a certain suffix $a_{t-n+1} a_{t-n+2} \cdots a_t$ of $a_1 a_2 \cdots a_t$
	and the algorithm has to output $f(a_{t-n+1} a_{t-n+2} \cdots a_t)$.
\end{itemize}
For many applications the sliding window model is more appropriate. Quite often data items in a stream are outdated
after a certain time, and the sliding window model is a simple way to model this. The typical application is the analysis of a
time series as it may arise in medical monitoring, web tracking, or financial monitoring.  In all these applications,
data items are usually no longer important after a certain time.
Two variants of the sliding window model can be found in the literature; see e.g. \cite{ArasuM04}:
\begin{itemize}
	\item {\em Fixed-size model:} The size of the sliding window is a fixed constant (the window size). In other words: at each time 
	      instant a new data value $a_i$ arrives and the oldest data value from the sliding window expires.
	\item {\em Variable-size model:} The sliding window $a_{t-n+1} a_{t-n+2} \cdots a_t$ is determined by an adversary. 
	      At every time instant the adversary can either remove the first data value from the sliding window (expiration of a value),
	      or add a new data value at the right end (arrival of a new value).	      
\end{itemize}
In the seminal paper of Datar et al.~\cite{DatarGIM02}, where the fixed-size sliding window model was introduced,
the authors show how to maintain the number of $1$'s in a sliding window of fixed size $n$ over the alphabet $\{0,1\}$
in space $\frac{1}{\varepsilon} \cdot \log^2 n$ if one allows a multiplicative error of $1\pm \varepsilon$.
A matching lower bound is proved as well in \cite{DatarGIM02}. For the upper bound, Datar et al. introduced a new
data structure called
exponential histograms. Histogram techniques
and variants have been used to approximate a large variety of statistical data over sliding windows. Let us mention the work
on computation of  the variance and $k$-median \cite{BabcockDMO03}, quantiles \cite{ArasuM04},
and entropy \cite{BravermanO07} over sliding windows. Other computational problems that have been considered
for the sliding window model include 
optimal sampling \cite{BravermanOZ12}, various pattern matching problems \cite{BreslauerG14,CliffordFPSS15,CliffordFPSS16,CliffordS16},
database querying (e.g. processing of join queries \cite{GolabO03}) and graph problems 
(e.g. checking for connectivity and computation of matchings, spanners, and
spanning trees~\cite{CrouchMS13}).
Further references on the sliding window model can be found in the surveys  \cite[Chapter~8]{Aggarwal07} and \cite{Braverman16}.
	
\subsection{Language recognition in the streaming model.}

A natural problem that has been suprisingly neglected for the streaming model is language recognition.
The goal is to check whether an input string belongs to a given language $L$.
Let us quote Magniez, Mathieu, and Nayak \cite{MagniezMN14}:
``Few applications [of streaming] have been made in the context of formal languages, which may have impact on massive data such as DNA sequences and large XML files. For instance, in the context of databases, properties decidable by streaming algorithm have been studied \cite{SegoufinV02,SegoufinS07}, but only in the restricted case of deterministic and constant memory space algorithms.'' 
For Magniez et al.~this was the starting point to study  language recognition in the streaming model. Thereby they
restricted their attention to the above mentioned standard streaming model. Note that 
in the standard model  the membership problem 
for a regular language is trivial to solve: One simply has to simulate a DFA on the stream and thereby only store
the current state of the DFA. In \cite{MagniezMN14} the authors present a randomized streaming
algorithm for the (non-regular) Dyck language $D_s$ with $s$ pairs of parenthesis that works in space $\mathcal{O}(\sqrt{n} \log n)$
and time $\mathrm{polylog}(n)$ per symbol.
Further investigations on streaming language recognition
for various subclasses of context-free languages can be found 
in  \cite{BabuLRV13,BabuLV10,FrancoisMRS16,KonradM13,KrebsLS11,SegoufinS07,SegoufinV02}.
Let us emphasize that all these papers exclusively deal with the standard streaming 
model. Language recognition problems for the sliding window model have been completely neglected so far.
This was the starting point for our previous paper \cite{GanardiHL16}.
	
\subsection{Querying regular languages in the sliding window model.}

As mentioned above, the membership problem for a regular language $L$ has a trivial constant space solution in the 
standard streaming model:
One simply simulates a DFA for $L$ on the data stream by storing the current state.
This solution does not work for the sliding window model.
The problem is the removal of the left-most symbol from the sliding window. 
In order to check whether the active window belongs to a certain language $L$ one 
has to know this first symbol in general. In such a case one has to store the whole window content 
using $\mathcal{O}(n)$ bits (where $n$ is the window size).
A simple regular language where this phenomenon arises is the language $a \{a,b\}^*$ of all words that start with $a$. The point is that by repeatedly
checking whether the sliding window content belongs to $a \{a,b\}^*$, one can recover the exact content of the sliding window, which implies
that every sliding window algorithm for testing membership in  $a \{a,b\}^*$  has to use $n$ bits of storage (where $n$ is the window size).
	
For a function $s(n)$ let $\Freg(s(n))$ be the class of all languages $L$ with the following property: For 
every window size $n$ there exists an algorithm that reads a data stream, uses only space $s(n)$ 
and correctly decides at every time instant whether the active window (the last $n$ symbols
from the stream) belongs to $L$. Note that this is a {\em non-uniform} model: for every window size $n$ we 
use a separate algorithm.  The class $\Vreg(s(n))$ of languages that have variable-size sliding window algorithm with space complexity 
$s(n)$ is defined similarly, see page \pageref{def-fixed-size} for details. 
Our main result from \cite{GanardiHL16} is a space trichotomy for regular languages: 
\begin{enumerate}[(i)]
	\item $\Vreg(o(n)) = \Freg(o(n)) = \Freg(\mathcal{O}(\log n)) =  \Vreg(\mathcal{O}(\log n))$
	\item $\Freg(o(\log n)) = \Freg(\mathcal{O}(1))$
	\item $\Vreg(o(\log n)) = \Vreg(\mathcal{O}(1)) =$ all trivial languages (empty and universal languages)
\end{enumerate}  

\noindent
Each of the three cases is characterized in terms of the
syntactic homomorphism and the left Cayley graph of the syntactic monoid of the regular language.
The precise characterizations are a bit technical; see \cite{GanardiHL16} for the details.

In this paper we continue our investigation of sliding-window algorithms for regular languages.
As a first contribution, we present very
natural characterizations of the above classes in (i) and (ii): 
The languages in (i) are exactly the languages  that are 
reducible with a Mealy machine (working from right to left) to a regular language of polynomial growth.
Note that the regular languages
of polynomial growth are exactly the bounded regular languages \cite{SzilardYZS92}. A language $L$ is {\em bounded} if $L \subseteq w_1^* w_2^* \cdots w_n^*$
for words $w_1, w_2, \ldots, w_n$.  In addition, we show that the class (i) is the Boolean closure of regular left ideals (regular languages $L$ with
$\Sigma^* L \subseteq L$) and regular length languages (regular languages where $|u| = |v|$ implies that $u \in L$ iff $v \in L$).
The class (ii) is characterized as the Boolean closure of suffix-testable languages (languages $L$ where membership in $L$ only depends
on a suffix of constant length) and regular length languages.
A natural example for the classes above is the problem of testing whether the sliding window contains a fixed pattern $w$ as a factor (as a suffix) since we can check membership of the left ideal $\Sigma^*w\Sigma^*$ (or of the suffix-testable language $\Sigma^*w$).
		
We also consider the sliding-window space complexity of  regular languages in a uniform setting, where the 
size $m$ (number of states) of an automaton for the regular language is also taken into account.
In \cite{GanardiHL16}, we asked whether for DFAs of size $m$ that accept languages in $\Freg(\mathcal{O}(\log n)) =  \Vreg(\mathcal{O}(\log n))$,
there exists a sliding-window streaming algorithm with
space complexity  $\mathrm{poly}(m) \cdot \log n$. Here, we give a negative answer by proving a lower bound of the form
$\Omega(2^m \cdot \log n)$. Moreover, we also show almost matching upper bounds.

Finally,  we prove that one can test in nondeterministic logspace and hence in deterministic polynomial time 
whether  for a given DFA $\A$ the language $L(\A)$ belongs to the above class (i)  (resp., (ii)).
For NFAs these problems become \PSPACE-complete.

\subsection{Related work.}	
In \cite{Fijalkow16} Fijalkow defines the online space complexity of a language $L$. His definition is equivalent
to the space complexity of the language $L$ in the standard streaming model described above.
Among other results,  Fijalkow presents a probabilistic automaton
$\mathcal{A}$ such that the language accepted by $\mathcal{A}$ (with threshold $1/2$) needs
space $\Omega(n)$ in the streaming model. 
	
Streaming a language $L$ in the standard model is also related to the concept of  automaticity \cite{ShallitB96}. 
For a language $L \subseteq \Sigma^*$, the automaticity $A_L$ of $L$ is the function $n \mapsto A_L(n)$, where
$A_L(n)$ is the minimal number of states of a DFA $\mathcal{A}$ such that for all words $w$ of length at most $n$: 
$w \in L$ if and only if $w \in L(\mathcal{A})$.
Clearly, every regular language $L$ has constant automaticity. Karp \cite{Karp67} proved that for every non-regular language $L$,
$A_L(n) \geq (n+3)/2$ for infinitely many $n$. This implies that for every non-regular language
$L$, membership checking in the standard streaming model is not possible in space $o(\log n)$.

\section{Preliminaries}

Throughout this paper we use 
$\log x$ as an abbreviation for  $\lfloor \log_2 x \rfloor$. Note that if $w_1, w_2, w_3, \ldots$ 
is the length-lexicographic enumeration of all words from $\{0,1\}^*$ then $|w_i| \leq \log i$.
We use the following well-known bounds for binomial coefficients, where $e$ is Euler's constant:
\[
	\bigg(\frac{n}{k}\bigg)^k  \leq \binom{n}{k} \leq \bigg(\frac{e \cdot n}{k}\bigg)^k, \quad \text{for all $1 \le k \le n$}.
\]
Assume that $u_1, \ldots, u_k \in \{0,1\}^+$ are non-empty bit strings of total length $n = \sum_{i=1}^k |u_i|$.
To encode the tuple $(u_1, \ldots, u_k)$ we use a simple block code: We encode each bit in $u_i$ except the first one
by the mapping $0 \mapsto 00$, $1 \mapsto 01$. The first bit in each $u_i$ is encoded by the 
mapping $0 \mapsto 10$, $1 \mapsto 11$. Then, the resulting bit strings are concatenated, which results in an encoding
with $2n$ bits. In the rest of the paper, we will use this encoding without mentioning it explicitly.
In fact, more succinct encodings exist.

Let $\Sigma^{\leq n} = \{  w \in \Sigma^* : |w| \leq n \}$.
A {\em prefix} of a word $w \in \Sigma^*$ is a word $u \in \Sigma^*$ with $w = uv$ for some $v \in \Sigma^*$.
The set of all prefixes of $w \in \Sigma^*$ is denoted by $\mathrm{Pref}(w)$.
For a language $L \subseteq \Sigma^*$ we define $\mathrm{Pref}(L) = \bigcup_{w \in L} \mathrm{Pref}(w)$
to be the set of prefixes of words in $L$.

The {\em reversal} of a word $x = a_1 \cdots a_n$ is defined as $x\rev = a_n \cdots a_1$
and the {\em reversal} of a language $L$ is $L\rev = \{ x\rev : x \in L\}$.
The {\em reversal} of a function $\tau : \Sigma^* \to \Gamma^*$ is defined as $\tau\rev(x) = \tau(x\rev)\rev$.
Thus, $\tau(u) = v$ if and only if $\tau\rev(u\rev) = v\rev$.

\subsection{Automata.}

We use standard definitions from automata theory.
A {\em nondeterministic finite automaton} (NFA) is a tuple $\A = (Q,\Sigma,I,\Delta,F)$
where $Q$ is a finite set of states, $\Sigma$ is an alphabet, $I \subseteq Q$ is the set of initial states,
$\Delta \subseteq Q \times \Sigma \times Q$ is the transition relation and $F \subseteq Q$ is the set of final states.
A {\em deterministic finite automaton} (DFA) $\A = (Q,\Sigma,q_0,\delta,F)$ has a single initial state $q_0 \in Q$ instead of $I$ and
a transition function $\delta \colon Q \times \Sigma \to Q$ instead of the transition relation $\Delta$.
A {\em deterministic automaton} has the same format as a DFA, except that the state set $Q$ is not required to be finite.
If $\A$ is deterministic, the transition function $\delta$ is extended to a function $\delta \colon Q \times \Sigma^* \to Q$
in the usual way and we define $\A(x) = \delta(q_0,x)$ for $x \in \Sigma^*$.
The language accepted by $\A$ is denoted by $L(\A)$.

The {\em Myhill-Nerode congruence} $\sim_L$ of a language $L \subseteq \Sigma^*$ is the equivalence relation on $\Sigma^*$
defined by $x \sim_L y$ if and only if 
\[
	\forall z \in \Sigma^* : xz \in L \iff yz \in L,
\]
which is a {\em right congruence} on $\Sigma^*$, i.e. $x \sim_L y$ implies $xz \sim_L yz$ for all $x,y,z \in \Sigma^*$.
For a word $x \in \Sigma^*$ the left quotient $x^{-1}L$ is $\{ z \in \Sigma^* : x z \in L \}$. Thus,
$x \sim_L y$ if and only if  $x^{-1}L = y^{-1}L$.
If $\A$ is a deterministic automaton for a language $L \subseteq \Sigma^*$,
then $\A(x) = \A(y)$ implies $x \sim_L y$.
Furthermore, $L$ is recognized by the deterministic automaton $\A = (Q,\Sigma,q_0,\delta,F)$
with state set $Q = \Sigma^*/{\sim_L}$, transition function $\delta([x]_{\sim_L},a) = [xa]_{\sim_L}$,
initial state $q_0 = [\varepsilon]_{\sim_L}$ and final states $F = \{ [x]_{\sim_L} : x \in L \}$,
which is the {\em minimal deterministic automaton} for $L$ (up to isomorphism).

For an NFA $\A$ we denote with $\A\det$ the corresponding deterministic power set automaton 
(restricted to those states that are reachable from the initial state) and with
$\A\rev$ the NFA obtained from $\A$ by reversing
all transitions and swapping the set of initial states and the set of final states. Moreover, we define
$\A\revdet = (\A\rev)\det$. Thus, 
$L(\A\rev) = L(\A\revdet) = L(\A)\rev$.
If an NFA $\A$ has $m$ states, then both $\A\det$ and $\A\revdet$ have at most $2^m$ states.

A language $L \subseteq \Sigma^*$ is {\em recognized} by a monoid $M$,
if there exists a homomorphism $h \colon \Sigma^* \to M$ and a set $F \subseteq M$ such that $h^{-1}(F) = L$.
The {\em syntactic congruence} $\equiv_L$ of $L$ is defined by $x \equiv_L y$ if and only if 
\[
	\forall u,v \in \Sigma^* : uxv \in L \iff uyv \in L.
\]
It refines $\sim_L$, i.e., $x \equiv_L y$ implies $x \sim_L y$.
The {\em syntactic monoid} of a language $L$ is the quotient monoid $\Sigma^* /{\equiv_L}$
and the mapping $h \colon \Sigma^* \to \Sigma^* /{\equiv_L}$, $h(x) = [x]_{\equiv_L}$ is the {\em syntactic homomorphism} of $L$.
It is known that a language is regular if and only if its syntactic monoid is finite.

\subsection{Streaming algorithms.}

A data stream is just a finite sequence of data values. We make the assumption that these data values are from 
a finite set $\Sigma$. Thus, a data stream is a finite word $w = a_1 a_2 \cdots a_m \in \Sigma^*$. 
A streaming algorithm reads the symbols of a data stream from left to right. At time instant 
$t$ the algorithm has only access to the symbol $a_t$ and the internal storage, which is encoded
by a bit string. The goal of the streaming algorithm is to compute a certain function $f \colon \Sigma^* \to A$ 
into some domain $A$, which means that at time instant $t$ the streaming algorithm outputs the value $f(a_1 a_2 \cdots a_t)$.
In this paper, we only consider the Boolean case $A = \{0,1\}$; in other words, the streaming algorithm tests
membership of a fixed language.
Furthermore, we abstract away from the actual computation and only analyze the space requirement. 
Formally, a {\em streaming algorithm} over $\Sigma$ is a deterministic (possibly infinite) automaton $\A = (S,\Sigma,s_0,\delta,F)$,
where the states are encoded by bit strings. We describe this encoding by
an injective function $\mathrm{enc} \colon S \to \{0,1\}^*$.
The {\em space function} $\mathrm{space}(\A,\cdot) \colon \Sigma^* \to \N$ specifies the space used by $\A$ on a certain input:
For $w \in \Sigma^*$ let 
$\mathrm{space}(\A,w) = \max \{ |\mathrm{enc}(\A(u))| : u \in \mathrm{Pref}(w)\}$.
We also say that $\A$ is a {\em streaming algorithm for} the accepted language $L(\A)$.

\section{Sliding window streaming models}

In the above streaming model, the output value of the streaming algorithm at time $t$ depends on the whole past $a_1 a_2 \cdots a_t$ 
of the data stream. However, in many practical applications one is only interested in the relevant part of the past. 
Two formalizations of ``relevant past'' can be found in the literature:
\begin{itemize}
	\item Only the suffix of $a_1 a_2 \cdots a_t$ of length $n$ is relevant. Here, $n$ is a fixed constant. 
	      This streaming model is called the {\em fixed-size sliding window model}. 
	\item The relevant suffix of $a_1 a_2 \cdots a_t$ is determined by an adversary. In this model, 
	      at every time instant the adversary can either remove the first symbol from the active window (expiration of a data value),
	      or add a new symbol at the right end (arrival of a new data value).
	      This streaming model is also called the {\em variable-size sliding window model}. 
\end{itemize}
In the following two paragraphs, we formally define these two models.

\subsection{Fixed-size sliding windows.}

Given a word $w = a_1 a_2 \cdots a_m \in \Sigma^*$ and a window length $n \ge 0$, we define $\last_n(w) \in \Sigma^n$ by
\[
	\last_n(w) =
	\begin{cases}
		a_{m-n+1} a_{m-n+2} \cdots a_m, & \text{if } n \le m, \\
		a^{n-m} a_1 \cdots a_m,         & \text{if } n > m,   
	\end{cases}
\]
which is called the {\em active window}.
Here $a \in \Sigma$ is an arbitrary symbol, which fills the initial window.
A sequence $\A = (\A_n)_{n \ge 0}$ is a {\em fixed-size sliding window algorithm} for a language $L \subseteq \Sigma^*$
if each $\A_n$ is a streaming algorithm for $\{ w \in \Sigma^* : \last_n(w) \in L \}$.
Its {\em space complexity} is the function $f_{\A} \colon \N \to \N \cup \{\infty\}$ where $f_{\A}(n)$ is the 
maximum encoding length of a state in $\A_n$. 

Note that for every language $L$ and every $n$ the language $\{ w \in \Sigma^* : \last_n(w) \in L \}$ is regular,
which ensures that $\A_n$ can be chosen to be a DFA and hence $f_{\A}(n) < \infty$ for all $n \ge 0$.
The trivial fixed-size sliding window algorithm for $L$ is the sequence 
$\B = (\B_n)_{n \ge 0}$, where $\B_n$ is the DFA with state set $\Sigma^n$ and transitions
$au \xrightarrow{b} ub$ for $a,b \in \Sigma$, $u \in \Sigma^{n-1}$.
States of $\B_n$ can be encoded with $\mathcal{O}(\log |\Sigma| \cdot n)$ bits.
By minimizing each $\B_n$, we obtain an {\em optimal fixed-size sliding window algorithm} 
$\A$ for $L$. Finally, we define $F_L(n) = f_{\A}(n)$.
Thus, $F_L$ is the space complexity of an optimal fixed-size sliding window algorithm for $L$.
Notice that $F_L$ is not necessarily monotonic. For instance, take $L = \{ au \colon u \in \{a,b\}^*, |u| \text{ odd} \}$. 
Then, we have $F_L(2n) \in \Theta(n)$ and  $F_L(2n+1) \in O(1)$.
The above trivial algorithm $\B$ yields $F_L(n) \in \mathcal{O}(n)$ for every language $L$.

Note that the fixed-size sliding window is a {\em non-uniform} model: for every window size we have
a separate streaming algorithm and these algorithms do not have to follow a common pattern.
Working with a non-uniform model makes lower bounds stronger.
In contrast, the variable-size sliding window model that we discuss next is a uniform model in the 
sense that there is a single streaming algorithm that works for every window length.

\subsection{Variable-size sliding windows.} 
For an alphabet $\Sigma$ we define the extended alphabet $\overline \Sigma = \Sigma \cup \{\downarrow\}$.
In the variable-size model the {\em active window} $\wnd(u) \in \Sigma^*$ for a stream $u \in \overline \Sigma^*$ is defined by
\begin{itemize}
	\item $\wnd(\varepsilon) = \varepsilon$
	\item $\wnd(ua) = \wnd(u) \cdot a$ for $a \in \Sigma$
	\item $\wnd(u \! \downarrow) = \varepsilon$ if $\wnd(u) = \varepsilon$
	\item $\wnd(u \! \downarrow) = v$ if $\wnd(u) = av$ for $a \in \Sigma$
\end{itemize}
A {\em variable-size sliding window algorithm} for a language $L \subseteq \Sigma^*$
is a streaming algorithm $\A$ for $\{ w \in \overline \Sigma^* : \wnd(w) \in L \}$.
Its {\em space complexity} is the function $v_{\A}\colon \N \to \N \cup \{\infty\}$ mapping each window length $n$ to the maximum number of bits used by $\A$
on inputs producing an active window of size at most $n$. 
Formally, it is the function
\[
	v_{\A}(n) = \max \{ \mathrm{space}(\A,u) : u \in \overline \Sigma^*, 
	|\wnd(v)| \le n \text{ for all } v \in \mathrm{Pref}(u) \},
\]
which is a monotonic function.\footnote{The definition of $v_\A(n)$ slightly deviates from the one given in \cite{GanardiHL16},
	namely $v'_{\A}(n) = \max \{ |\mathrm{enc}(\A(u))| : u \in \overline \Sigma^*, |\wnd(u)| = n \}$. 
	One easily sees that $v_\A(n) = \max_{k \le n} v'_{\A}(k)$ and hence $v_\A(n) = v'_\A(n)$
	for monotonic functions $v'_\A(n)$.}

\begin{lem}
	\label{lem:optimal-vs}
	For every language $L \subseteq \Sigma^*$ there exists a variable-size sliding window algorithm $\A$
	such that $v_{\A}(n) \leq v_{\B}(n)$ for every variable-size sliding window algorithm $\B$ for $L$ and every $n$.
\end{lem}

\begin{proof}
	Let $\A = (S,\overline\Sigma,s_0,\delta,F)$ be the minimal deterministic automaton for $\{ w \in \overline \Sigma^* : \wnd(w) \in L \}$.
	The state set $S$ can be finite or infinite.  
	It has the property that the active window determines the current state,
	i.e. $\A(x) = \A(\wnd(x))$ for all $x \in \overline \Sigma^*$.
	For $n \ge 0$ let $S_n \subseteq S$ be the set of states reachable in $\A$ from the initial state $s_0$
	by a word over $\overline \Sigma$ of length at most $n$.
	By the aforementioned property all words $x \in \overline \Sigma^*$ with $|\wnd(x)| \le n$
	lead to a state $\A(x) = \A(\wnd(x)) \in S_n$. 
	Now one can define an encoding such that the space complexity of $\A$ is $\log |S_n|$:
	Define an enumeration of $S$ by starting with $s_0$, then listing (in any order) all states  from $S_1 \setminus S_0$,
	followed by the states from $S_2 \setminus S_1$, and so one.
	Then we encode the $i$-th state from this list by the $i$-th bit string in length-lexicographical order.
							
	Now let $\B$ be any variable-size sliding window algorithm for $L$.
	Let $T_n$ be the set of states reachable in $\B$ from the initial state by a word of length at most $n$.
	By reading a word of length at most $n$, the window length never exceeds $n$.
	Therefore, the encoding length of any $t \in T_n$ is bounded by $v_\B(n)$, which implies $|T_n| \leq 2^{v_\B(n)+1}-1$.
	We get  $\log |T_n| \le v_{\B}(n)$.
	Since $\A$ is minimal we have $|S_n| \le |T_n|$ and therefore $v_{\A}(n) \leq v_{\B}(n)$.
\end{proof}
We define $V_L(n) = v_{\A}(n)$, where $\A$ is a space {\em optimal variable-size sliding window algorithm} for $L$
from \cref{lem:optimal-vs}.
Since any algorithm in the variable-size model yields an algorithm in the fixed-size model, we have $F_L(n) \le V_L(n)$.

\subsection{Space complexity classes and closure properties.}
For a function $s \colon \N \to \N$ we define the classes $\F(s)$ \label{def-fixed-size} and $\V(s)$
of all languages $L \subseteq \Sigma^*$ which have a fixed-size (variable-size, respectively) sliding window algorithm
with space complexity bounded by $s(n)$. For a class $\mathcal{C}$ of functions 
(here, it will be always an $\mathcal{O}$-class, $\Theta$-class  or $o$-class) we define 
$\mathsf{X}(\mathcal{C}) = \bigcup_{s \in \mathcal{C}}\mathsf{X}(s)$ for $\mathsf{X} \in \{ \F,\V \}$.


Several times we will make use of the simple fact that for both the fixed-size and the variable-size model,
space classes form a Boolean algebra:
\begin{lem}
	\label{lem:boolean}
	Let $X \in \{F,V\}$.
	If $L \subseteq \Sigma^*$ is a Boolean combination of languages $L_1, \ldots, L_k \subseteq \Sigma^*$,
	then $X_L(n) \le 2 \sum_{i=1}^k X_{L_i}(n)$.
	In particular, for any function $s(n)$, the classes $\F(\mathcal{O}(s))$ and $\V(\mathcal{O}(s))$ form  Boolean algebras.
\end{lem}

\begin{proof}
	Run the sliding window-algorithms for $L_1, \dots, L_k$ in parallel and encode 
	the tuple of $k$ states by a single bit string.
	The output bits of the individual algorithms determine the output of the total algorithm.
\end{proof}
A {\em Mealy machine} $\mathcal{M} = (Q,\Sigma,\Gamma,q_0,\delta)$ consists of a finite set of states $Q$,
an input alphabet $\Sigma$, an output alphabet $\Gamma$, an initial state $q_0 \in Q$ and the transition function
$\delta \colon Q \times \Sigma \to Q \times \Gamma$.
For every $q \in Q$ the machine computes a length-preserving transduction $\tau_q \colon \Sigma^* \to \Gamma^*$
in the usual way: $\tau_q(\varepsilon) = \varepsilon$ and if $\delta(p,a) = (q,b)$ then $\tau_p(au) = b \, \tau_q(u)$.
We call $\tau_{q_0}\rev$ the {\em $\leftarrow$-transduction} computed by $\mathcal{M}$.
Thus, a $\leftarrow$-transduction is computed by a Mealy machine that works on an input word 
from right to left.
If $L$ is regular and $\tau$ is a $\leftarrow$-transduction, then $\tau(L)$ and $\tau^{-1}(L)$ are regular as well.
A $\leftarrow$-transduction $\tau$ is called a {\em $\leftarrow$-reduction}
from $K \subseteq \Sigma^*$ to $L \subseteq \Gamma^*$ if $x \in K$ if and only if $\tau(x) \in L$ for all $x \in \Sigma^*$.

\begin{lem}
	\label{lem:left-reduction}
	Let $X \in \{F,V\}$.
	If $K$ is $\leftarrow$-reducible to $L$ via a Mealy machine with $d$ states,
	then $X_K(n) \le 2d \cdot X_L(n)$.
	In particular, for any function $s(n)$ the classes $\F(\mathcal{O}(s))$ and $\V(\mathcal{O}(s))$ are closed under $\leftarrow$-reductions.
\end{lem}
	
\begin{proof}
	We only give the proof for the variable-size model; analogous arguments hold for the fixed-size model.
	Let $\A$ be an optimal variable-size sliding window algorithm for $L$. 
	Recall from the proof of \cref{lem:optimal-vs} that $\A(w) = \A(\wnd(w))$ for all streams $w \in \overline \Sigma^*$.
	Let $\mathcal{M} = (Q,\Sigma,\Gamma,q_0,\delta)$ be a Mealy machine
	such that $\tau\rev_{q_0}$ is a $\leftarrow$-reduction from $K$ to $L$.
	Let $Q = \{q_0, \dots, q_{d-1} \}$ be the state set of $\mathcal{M}$.
	
	We claim that there exists a sliding window algorithm $\B$ which
	given an input stream $w \in \overline \Sigma^*$ maintains an encoding of the tuple
	\[ \B(w) = (\A(\tau_{q_0}\rev(\wnd(w))), \ldots, \A(\tau_{q_{d-1}}\rev(\wnd(w)))). \]
	\begin{itemize}
		\item On input $\downarrow$ we can compute
		      \[
		      	\B(w \! \downarrow) = (\A(\tau_{q_0}\rev(\wnd(w)) \! \downarrow), \ldots, \A(\tau_{q_{d-1}}\rev(\wnd(w)) \! \downarrow))
		      \]
		      from $\B(w)$.
		\item Given an input symbol $a \in \Sigma$, compute $\delta(q_i,a) = (p_i,b_i)$ for all $0 \le i \le d-1$.
		      Since
		      \[
		      	\tau_{q_i}\rev(\wnd(wa)) = \tau_{p_i}\rev(\wnd(w)) \; b_i
		      \]
		      we can compute
		      \[
		      	\B(wa) = (\A(\tau_{p_0}\rev(\wnd(w)) \; b_0), \ldots, \A(\tau_{p_{d-1}}\rev(\wnd(w)) \; b_{d-1})).
		      \]
		\item The active window belongs to $L$ if and only if $\A(\tau_{q_0}\rev(\wnd(w)))$ is final in $\A$.
	\end{itemize}
	The above variable-size sliding window algorithm has space complexity $2 d \cdot V_L(n)$:
	If $w =  \overline \Sigma^*$ is such that $|\wnd(u)| \le n$ for all $u \in \mathrm{Pref}(w)$,
	then also $|\tau_{q_i}\rev(\wnd(u))| \le n$ for all $u \in \mathrm{Pref}(w)$, $0 \leq i \leq d-1$. Hence,
	$|\A(\tau_{q_i}\rev(\wnd(u)))| \le V_L(n)$ for all $u \in \mathrm{Pref}(w)$, $0 \leq i \leq d-1$.
	Thus, every tuple $\B(a_1 \cdots a_t)$ can be encoded with at most $2 d \cdot V_L(n)$ bits.		
\end{proof}

\subsection{Space trichotomy for regular languages.}

In \cite{GanardiHL16} we proved a trichotomy theorem on sliding window algorithms for regular languages.
We identified a partition of the class of regular languages into three classes
which completely characterize the sliding window space complexity in both the fixed-size and the variable-size model.
One can easily see that the syntactic monoid of a language does not determine its space complexity:
$\mathbb{Z}_2$ is the syntactic monoid of both languages $K = $ ``even length'' and $L = $ ``even number of $a$'s''
over $\{a,b\}$ but $\V_{K}(n) = \mathcal{O}(\log n)$ whereas $\F_{L}(n) = \Theta(n)$.
The definition of the mentioned three classes is given in terms of the
syntactic homomorphism and the left Cayley graph
of the syntactic monoid of the regular language, see \cite{GanardiHL16}.

For $\mathsf{X} \in \{ \F, \V\}$ and a class $\mathcal{C}$ of functions 
we abbreviate $\mathsf{X}(\mathcal{C}) \cap \mathsf{REG}$ by $\mathsf{X}_{\mathsf{reg}}(\mathcal{C})$,
where $\mathsf{REG}$ is the class of all regular languages.

\begin{thm}[\cite{GanardiHL16}]\label{thm:trichotomy}
	The following holds:
	\begin{itemize}
		\item $\Vreg(o(n)) = \Freg(o(n)) = \Freg(\mathcal{O}(\log n)) =  \Vreg(\mathcal{O}(\log n))$
		\item $\Freg(o(\log n)) = \Freg(\mathcal{O}(1))$
		\item $\Vreg(o(\log n)) = \Vreg(\mathcal{O}(1)) =$ all trivial languages (empty and universal languages)
	\end{itemize}  
\end{thm}

Strictly speaking, \cite[Theorem~7]{GanardiHL16} only claims $\V_L(n) \notin O(1)$ for all languages $\emptyset \subsetneq L \subsetneq \Sigma^*$.
However, the proof of \cite[Theorem~7]{GanardiHL16} does imply the stronger bound $\V_L(n) \notin o(\log n)$.
This statement will also be reproved in the following section.

Let us comment on a subtle point. When making statements about the space complexity functions
$V_L(n)$ and $F_L(n)$ it is in general important to fix the underlying alphabet. For instance
according to point (iii) from Theorem~\ref{thm:trichotomy} we have $V_{L}(n) \in \mathcal{O}(1)$ for the language $L = \{a\}^*$
if the underlying alphabet is $\{a\}$. On the other hand, if the underlying alphabet is $\{a,b\}$ then 
$V_L(n) \not \in \mathcal{O}(1)$ (in fact, $L$ then belongs to $\Vreg(\Theta(\log n))$).

\section{Variable-size space complexity and language growth}

In this section we reprove the space trichotomy (\cref{thm:trichotomy}) for the variable-size model.
For this we relate  the function $V_L(n)$ to the growth of a certain derived language and then use
the well known results about the growth of regular languages.  
We need the following definition.
For a language $L \subseteq \Sigma^*$ define the mapping $\psi_L \colon \Sigma^*  \to (\Sigma^*/{\sim_L})^*$ by:
\[ \psi_L(a_1 \cdots a_n)  = [a_1 \cdots a_n]_{\sim_L} [a_2 \cdots a_n]_{\sim_L} \cdots [a_n]_{\sim_L}. \]
Notice that $\psi_L$ is a length-preserving mapping from $\Sigma^*$ to the set of words over
the alphabet $\Sigma^*/{\sim_L}$.
Although $\Sigma^*/{\sim_L}$ may be infinite (namely for non-regular $L$),
the image $\psi_L(\Sigma^{\le n})$ has at most $|\Sigma|^{n+1}-1$ elements for each $n \ge 0$.

\begin{thm}
	\label{lem:psi-L}
	For every language $\emptyset \subsetneq L \subsetneq \Sigma^*$ 
	we have $V_L(n) = \log |\psi_L(\Sigma^{\le n})|$.
\end{thm}
	
\begin{proof}
	We first exhibit a variable-size sliding window algorithm $\A$ with space complexity $\log |\psi_L(\Sigma^{\le n})|$.
	The idea is that on input $w \in \overline \Sigma^*$ the algorithm $\A$ is in state $\A(w) = \psi_L(\wnd(w))$.
	Consider an active window $a_1 \cdots a_n \in \Sigma^*$. Three observations are crucial:
	\begin{itemize}
		\item The state $\psi_L(a_2 \cdots a_n)$ can be obtained from the state $\psi_L(a_1 \cdots a_n)$ by removing 
		      the first $\sim_L$-class  $[a_1 \cdots a_n]_{\sim_L}$.
		\item From the state $\psi_L(a_1 \cdots a_n)$ and a symbol
		      $a \in \Sigma$ one can obtain the state $\psi_L(a_1 \cdots a_n a) = 
		      [a_1 \cdots a_n a]_{\sim_L} [a_2 \cdots a_n a]_{\sim_L} \cdots [a_n a]_{\sim_L} [a]_{\sim_L}$,
		      since $\sim_L$ is a right-congruence.
		\item The first $\sim_L$-class in $\psi_L(a_1 \cdots a_n)$ determines whether $a_1 \cdots a_n \in L$.
	\end{itemize}
	These remarks define a variable-size sliding window algorithm for $L$ with state set $\psi_L(\Sigma^*)$.
	It remains to define the binary encoding of the states, This is done similarly to 
	the proof of \cref{lem:optimal-vs}:
	List $\psi_L(\Sigma^*)$ by starting with $\psi_L(\varepsilon) = \varepsilon$, followed by all states from 
	$\psi_L(\Sigma)$ (in any order), followed by all states from $\psi_L(\Sigma^2)$, and so on. The $i$-th 
	state in this list is encoded by the $i$-th bit string in length-lexicographical order.
	Under this encoding the above  variable-size sliding window algorithm has space complexity $\log |\psi_L(\Sigma^{\le n})|$.
			
	Conversely, consider a variable-size sliding window algorithm $\A$ for $L$
	with space complexity $v_{\A}(n)$.
	We have to show that $v_\A(n) \geq \log |\psi_L(\Sigma^{\le n})|$.
	Let $x = a_1 a_2 \cdots a_m \in \Sigma^*$ be an input word of length $m \leq n$.
	Notice that $|\mathrm{enc}(\A(x))| \leq v(m) \leq v(n)$ by the monotonicity of $v$.
							
	We first show that $\A(x)$ determines $m=|x|$.
	Assume that $\varepsilon \in L$ (the case that $\varepsilon \notin L$ is analog),
	and let $y \notin L$ where $|y|$ is chosen minimally.
	Starting from $\A(x)$ we read $y$ into $\A$, followed by an infinite sequence of $\downarrow$. We obtain a run
	\[
		\A(x) \xrightarrow{y} s_0 \xrightarrow{\downarrow} s_1 \xrightarrow{\downarrow} s_2 \xrightarrow{\downarrow} s_3 \xrightarrow{\downarrow} \dots
	\]
	where $s_m$ is not final and for all $i > m$ the state $s_i$ is final, by minimality of $|y|$.
	Clearly this run determines $m$.
			
	We now show that $\A(x)$ uniquely determines $\psi_L(a_1 \cdots a_m)$.
	By the above argument, we know that $\A(x)$ determines the window length $m$.
	Furthermore, $\A(x)$ determines every equivalence class $[a_k \cdots a_m]_{\sim_L}$ for $1 \le k \le m$: 
	Starting from $\A(x)$ we read $k-1$ times $\downarrow$ into $\A$. Then, the active window is  $a_k \cdots a_m$.
	We can determine the left quotient $(a_k \cdots a_m)^{-1}L = \{ z \in \Sigma^* : a_k \cdots a_m z \in L \}$
	by reading each word $z$ into $\A$ and testing whether $a_k \cdots a_m z \in L$.
	The left quotient in turn determines $[a_k \cdots a_m]_{\sim_L}$.
							
	To sum up, we have shown that every value $\psi_L(x)$ for $x \in \Sigma^{\leq n}$ can be encoded by a bit
	string of length at most $v(n)$, namely $\mathrm{enc}(\A(x))$.
	Since there are $|\psi_L(\Sigma^{\le n})|$ such values, it follows that $2^{v(n)+1}-1 \geq |\psi_L(\Sigma^{\le n})|$,
	which implies $v(n) \geq \log |\psi_L(\Sigma^{\le n})|$.
\end{proof}
Note that \cref{lem:psi-L} does not hold for $L = \emptyset$ or $L = \Sigma^*$. In these cases, we have $V_L(n) = 0$ and
$\log |\psi_L(\Sigma^{\le n})| = \log (n+1)$.

We can use \cref{lem:psi-L} to reprove the space trichotomy for regular languages in the variable-size 
sliding window model. For this, we need the following simple lemma:
	
\begin{lem}
	\label{lem:psi-regular}
	If $L \subseteq \Sigma^*$ is regular, then $\psi_L$ is a $\leftarrow$-transduction.
	In particular, $\psi_L(\Sigma^*)$ and $\psi_L(L)$ are regular.
	Furthermore $\psi_L$ is a $\leftarrow$-reduction from $L$ to $\psi_L(L)$.
\end{lem}

\begin{proof}
	Let $h \colon \Sigma^* \to M$ be the syntactic homomorphism of $L$ into the syntactic monoid $M$ of $L$.
	Since the syntactic congruence refines the Myhill-Nerode congruence, there exists a function
	$\nu \colon M \to \Sigma^*/{\sim_L}$ such that $[x]_{\sim_L} = \nu(h(x))$ for all $x \in \Sigma^*$.
	Define the Mealy machine with the state set $M$ and transitions
	\[
		\delta(m,a) = (h(a) \cdot m, \nu(h(a) \cdot m))
	\]
	for all $m \in M$, $a \in \Sigma$.
	This Mealy machine computes the $\leftarrow$-transduction $\psi_L$.
		
	If $\psi_L(x) = \psi_L(y)$ then either both or none of the words $x,y$ belong to $L$.
	This proves that $\psi_L$ is indeed a reduction from $L$ to $\psi_L(L)$.
\end{proof}
The {\em growth} of a language $L \subseteq \Sigma^*$ is the function $g(n) = |\{ x \in L : |x| \le n\}|$.
Since the growth of every regular  language is either $\Theta(n^d)$ for some integer $d \ge 0$ or 
$\Omega(r^n)$ for some $r > 1$ \cite[Section~2.3]{GKRS10},
\cref{lem:psi-L,lem:psi-regular} reprove the trichotomy theorem for variable-size windows:
For every regular language $L$, $V_L(n)$ is either in $\mathcal{O}(1)$, $\Theta(\log n)$ or $\Theta(n)$.
Furthermore, since $|\psi_L(\Sigma^{\le n})| \ge n+1$ we have $V_L(n) \in \Omega(\log n)$ for every non-trivial language $L$.

\begin{thm}
	If $L \subseteq \Sigma^*$ has growth $g(n)$,
	then $F_L(n) \in \mathcal{O}(\log g(n) +  \log n)$. 
\end{thm}
     
\begin{proof}
	Let $n \ge 0$ be a window size and let $w_1, \dots, w_m$ be an arbitrary enumeration of $L \cap \Sigma^n$
	where $m \le g(n)$.
	Assume that $w = a_1 \cdots a_n \in \Sigma^*$ is the active window.
	The algorithm stores the longest suffix $v = a_i \cdots a_n$ of $w$
	such that $v$ is a prefix of a word $w_j \in L \cap \Sigma^n$.
	Notice that $v$ can be encoded by the binary encoded number $j$ using $\log g(n)$ bits
	and the binary encoded number $i$ using $\log n$ bits.
	Of course, there may exist several words $w_j$ having $v$ as a prefix;
	in this case the concrete choice of $w_j$ does not matter.
	This information clearly suffices to check whether the active window belongs to $L$.
	Moreover, we can update the information:
	If $a_{n+1} \in \Sigma$ is the next symbol from the stream, then we distinguish the following cases:
        	\begin{itemize}
        		\item If $i > 1$ and $a_i \cdots a_n a_{n+1}$ is a prefix of a word from
        		$L \cap \Sigma^n$, say $w_{j'}$, $1 \leq j' \leq m$, then we replace $i,j$ by
        		$i-1,j'$.
        		\item Otherwise let $i < i' \le n+1$ be minimal such that $a_{i'} \cdots a_n a_{n+1}$ is a prefix of a word from
        		$L \cap \Sigma^n$, say $w_{j'}$, $1 \leq j' \leq m$. We 
        		replace $i,j$ by  $i',j'$.
        	\end{itemize} 
        	The correctness of this algorithm is straightforward.
     \end{proof}

\section{Logspace sliding-window algorithms}
	
In this section, we give a new and more natural characterization of languages in $\Vreg(\mathcal{O}(\log n))$.  Moreover, we analyze the 
influence of the size of the automaton on the $\mathcal{O}$-constant. In \cite{GanardiHL16} we gave the space bound 
$\mathcal{O}(m^m \cdot (m \cdot \log(m) + \log(n)))$ if the regular language is given by a DFA with $m$ states.
Below, we improve this bound to $\mathcal{O}(2^m \cdot m \cdot \log(n))$.

Let $\B = (Q,\Sigma,q_0,\delta,F)$ be a DFA. A {\em strongly connected component} (SCC for short)
of $\B$ is an inclusion-maximal subset $C \subseteq Q$ such that
for all $p,q \in C$ there exist words $u,v \in \Sigma^*$ such that $\delta(p,u)=q$ and $\delta(q,v) = p$.
The crucial property that enables logspace sliding-window algorithms is captured by the following definition:

\begin{defn}
	Let $\B = (Q,\Sigma,q_0,\delta,F)$ be a DFA. An SCC $C \subseteq Q$ is {\em well-behaved} if for all $q \in C$
	and $u,v \in \Sigma^*$ with $|u|=|v|$ and $\delta(q,u), \delta(q,v) \in C$ we have:
	$\delta(q,u) \in F$ if and only if $\delta(q,v) \in F$.
	If every SCC in $\B$ which is reachable from $q_0$ is well-behaved, then $\B$ is called {\em well-behaved}.
\end{defn}
It turns out that $L \in \Vreg(\mathcal{O}(\log n))$ if and only if $L\rev$ can be accepted
by a well-behaved DFA, and we will prove this fact below. 
Thereby we determine the dependence of the constant in the $\mathcal{O}(\log n)$ bound with respect to the size of 
an automaton (DFA or NFA) for $L$.

Let $\B$ be a well-behaved DFA and let $\rho$ be a run in $\B$, which does not necessarily start in the initial state.
Let $C_1, \ldots, C_k$ be the sequence of pairwise different SCCs that are visited by $\rho$ in that particular order.
The {\em path summary} of $\rho$ is the sequence $(p_1,\ell_1,p_2,\ell_2, \dots, p_k, \ell_k)$ where
$p_i$ is the first state in $C_i$ visited by $\rho$, and $\ell_i \ge 0$ is the number of symbols read in $\rho$
from the first occurrence of $p_i$ until the first state from $C_{i+1}$ (or until the end for $p_k$).
The number of different path summaries of runs of length $n$  in a DFA $\B$ with $m$ states can be bounded by
($e$ is Euler's constant)
\begin{equation} \label{number-summaries}
	m^m \cdot \binom{n+m-1}{m-1} \leq m^m \cdot \binom{n+m}{m} \leq m^m \cdot \bigg(\frac{e \cdot (n+m)}{m}\bigg)^m \leq e^m \cdot (n+m)^m .
\end{equation}
Here, (i) $m^m$ is the number of sequences of $m$ states (we can repeat the last state in a path summary so that we have exactly $m$ states)
and (ii) $\binom{n+m-1}{m-1}$ is the number of ordered partitions of $n$ into $m$ summands.

\begin{thm}
	\label{prop:uniform-algorithm}
	Let $L \subseteq \Sigma^*$ be regular and let $\A$ be a finite automaton for $L$ with $m$ states.
	Assume that $\B = \A^{\R\D}$ is well-behaved. There are constants $c_m, d_m$ that only depend on $m$ such that 
	the following holds:
	\begin{itemize}
		\item If $\A$ is a DFA then $V_L(n) \leq (2^m \cdot m +1) \cdot \log n  \ + \ c_m$ for $n$ large enough.
		\item If $\A$ is an NFA then $V_L(n) \leq (4^m +1) \cdot \log n  \ + \ d_m$ for $n$ large enough.
	\end{itemize}
\end{thm}

\begin{proof}
	A set $D \subseteq \Sigma^*$ {\em distinguishes} $L$
	if for all $x,y \in \Sigma^*$ with $x \not \sim_L y$ there exists $z \in D$ such that exactly
	one of the words $xz$ and $yz$ belongs to $L$.
	If $\A$ is a DFA with $m$ states, then there are at most $m$ distinct left quotients $x^{-1} L$.
	Since every family of $m$ sets has a distinguishing set of size at most $m-1$ \cite{PPR97}, we get a
	set $D$ of size at most $m-1$  that distinguishes $L$.
	If $\A$ is an NFA with $m$ states, we can clearly choose $|D| \le 2^m-1$ by determinizing $\A$.
		
	For a window content $w = a_1 \cdots a_n$ we define a 0-1-matrix $A_w \colon D \times \{1, \dots, n\} \to \{0,1\}$ by
	$A_w(z,i) = 1$ iff $a_i \cdots a_n z \in L$.
	Notice that the $i$-th column $A_w(\cdot,i)$ determines $[a_i \cdots a_n]_{\sim_L}$, and vice versa, for all $1 \le i \le n$.
	In particular, the matrix $A_w$ determines $\psi_L(w)$ and vice versa. Thus, $|\psi_L(\Sigma^{\leq n})| = |\{ A_w \colon w \in \Sigma^{\leq n }\}|$.
	By \cref{lem:psi-L}, it therefore suffices to bound $|\{ A_w \colon w \in \Sigma^{\leq n }\}|$.
				
	We can encode each row $A_w(z,\cdot)$ of $A_w$ succinctly as follows. Consider one row indexed by $z \in D$.
	Let $\rho_z$ be the run of $\B$ on the word $(wz)\rev$ and $\tilde \rho_z$ be
	the subrun of $\rho_z$ which only reads the suffix $w\rev$ of $(wz)\rev$.
	One can reconstruct $A_w(z,\cdot)$ from the path summary of $\tilde \rho_z$.
	Thus $A_w$ can be encoded by $|D|$ many path summaries. With \eqref{number-summaries}
	and the fact that $\B$ has at most $2^m$ states, we get the bound
	\[ 
		|\{ A_w \colon w \in \Sigma^{\leq n}\}| \leq \sum_{i=0}^n e^{2^m |D|} \cdot (i+2^m)^{2^m |D|} 
		\leq
		(n+1) \cdot e^{2^m |D|} \cdot (n+2^m)^{2^m |D|} .
	\]
	Hence, 
	for the DFA case (where $|D| \leq m-1$) we have
	\begin{eqnarray*}
          V_L(n) &=& \log |\psi_L(\Sigma^{\leq n})| \\
          &\leq & \log (n+1) + 2^m  \cdot m \cdot ( \log e + \log (n + 2^m) ) \\
          &\leq &  (2^m \cdot m +1) \cdot \log n  \ + \ c_m
         \end{eqnarray*}
	for $n$ large enough, where $c_m$ can be chosen as $1 + 2^m \cdot m \cdot \log e + m^2 \cdot 2^m$.
	The calculation for the NFA case (where $|D| \leq 2^m-1$) is analogous.
\end{proof}
Finally, we show a linear space lower bound for the case that the reversal of $L$ is recognized by a non-well-behaved DFA.

\begin{thm} \label{prop-linear-lower-bound}
	Let $L \subseteq \Sigma^*$ be a regular language and $\B$ be a DFA which recognizes $L\rev$ and which is not well-behaved.
	Then $V_L(n) \in \Omega(n)$ and $F_L(n) \in \mathcal{O}(n) \setminus o(n)$.
\end{thm}
	
\begin{proof}
	Let $q_0$ be the initial state of $\B$.
	Since $\B$ is not well-behaved, there are states $p, p_0, p_1$ and words $u, u_0, v_0, u_1, v_1 \in \Sigma^*$
	such that $|u_0| = |v_0|$, $p_0$ is not final, $p_1$ is final and
	$q_0 \xrightarrow{u} p$, $p \xrightarrow{u_0} p_0 \xrightarrow{v_0} p$ and $p \xrightarrow{u_1} p_1 \xrightarrow{v_1} p$.
	We can ensure that $|u_1| = |v_1|$: If $k = |u_0v_0|$ and $\ell = |u_1v_1|$, we replace $v_0$ by $v_0 (u_0v_0)^{\ell-1}$
	and $v_1$ by $v_1 (u_1v_1)^{k-1}$.
							
	For any $\alpha = \alpha_1 \cdots \alpha_n \in \{0,1\}^*$ we define
	the word
	\[
		w(\alpha) = u \, u_{\alpha_1}v_{\alpha_1} \cdots u_{\alpha_n}v_{\alpha_n}.
	\]
	Notice that the length of $w(\alpha)$ is $|u| + k \ell |\alpha| \in \mathcal{O}(|\alpha|)$.
	Let $\alpha \neq \beta$ be two bit strings of length $n$ which differ in position $i$,
	say $\alpha_i = 1$ and $\beta_i = 0$.
	Then $\B$ accepts $u \, u_{\alpha_1}v_{\alpha_1} \cdots u_{\alpha_i}$
	but rejects $u \, u_{\beta_1}v_{\beta_1} \cdots u_{\beta_i}$, which are prefixes of $w(\alpha)$
	and $w(\beta)$, respectively, of the same length.
	In particular, $\psi_L(w(\alpha)\rev) \neq \psi_L(w(\beta)\rev)$.
	Therefore, for any $n \ge 0$, the language $\psi_L(\Sigma^*)$ contains at least $2^n$ words of length $\mathcal{O}(n)$.
	By \cref{lem:psi-L} and monotonicity of $V_L(n)$, this implies $V_L(n) = \Omega(n)$.
	By \cref{thm:trichotomy} we also know that $F_L(n) \in \mathcal{O}(n) \setminus o(n)$.
\end{proof}
From \cref{prop:uniform-algorithm} and \cref{prop-linear-lower-bound} we obtain:
\begin{cor}
	Let $\mathsf{X} \in \{\F,\V\}$.
	A regular language $L \subseteq \Sigma^*$ belongs to $\mathsf{X}(\mathcal{O}(\log n))$
	if and only if $L^\R$ is recognized by a well-behaved DFA.
\end{cor}

\subsection{Alternative Characterizations of $\Vreg(\mathcal{O}(\log n))$}

In the following we will give two further very natural characterizations of the languages in
$\Vreg(\mathcal{O}(\log n))=\Freg(\mathcal{O}(\log n))$ that we will also need in \cref{sec-computing-space}.

A language $L \subseteq \Sigma^*$ is called a {\em left ideal (right ideal)} if $\Sigma^* L \subseteq L$ ($L \Sigma^* \subseteq L$).
A language $L \subseteq \Sigma^*$ is called a {\em length language} if for all $n \in \mathbb{N}$,
either $\Sigma^n\subseteq L$ or $L \cap \Sigma^n = \emptyset$. Clearly,  $L$ is a length language iff 
$L\rev$ is a length language, and $L$ is left ideal iff $L\rev$ is a right ideal.
In this section we will prove the following theorem.

\begin{thm} \label{thm-characterizations}
	Let $L \subseteq \Sigma^*$ be regular.
	The following statements are equivalent:
	\begin{enumerate}
		\item $L \in \F(\mathcal{O}(\log n))$
		\item $L \in \V(\mathcal{O}(\log n))$
		\item $L\rev$  is recognized by a well-behaved DFA.
		\item $L$ is $\leftarrow$-reducible to a regular language of polynomial growth.
		\item $L$ is a Boolean combination of regular left ideals and regular length languages.
	\end{enumerate}	
\end{thm}

\noindent
The equivalence of points 1.~and 2.~was already shown in \cite{GanardiHL16}, and the equivalence
of 2.~and 3.~was shown in the last section.
The implication from 2.~to 4.~follows from \cref{lem:psi-L,lem:psi-regular}.
In the rest of the section, we prove the directions from 5.~to 3., and from 4.~to 5.

We start with two simple observations, which prove the direction from 5.~to 3.

\begin{lem}
	If a regular language $L \subseteq \Sigma^*$ is a right ideal or a length language,
	then the minimal DFA for $L$ is well-behaved.
\end{lem}

\begin{proof}
	If $\A$ is the minimal DFA for a length language then for all states 
	$q$ and all $u,v \in \Sigma^*$ with $|u|=|v|$, we have:
	$\delta(q,u) \in F$ if and only if $\delta(q,v) \in F$.
	
	If $\A$ is the minimal DFA for a right ideal, then for all final states
	$q$ and all $u \in \Sigma^*$, the state $\delta(q,u)$ is final as well.
	Hence, for every SCC $C$ either all states of $C$ are final or all states
	of $C$ are non-final.
\end{proof}

\begin{lem}
	The class of languages $L \subseteq \Sigma^*$ recognized by well-behaved DFAs
	is closed under Boolean operations.
\end{lem}

\begin{proof}
	If $\A$ is well-behaved then the complement automaton $\overline \A$ is also well-behaved.
	Given two well-behaved DFAs $\A_1,\A_2$,
	we claim that the product automaton $\A_1 \times \A_2$ recognizing the intersection language is also well-behaved.
	Consider an SCC $S$ of $\A_1 \times \A_2$ which is reachable from the initial state and let $(p_1,p_2), (q_1,q_2), (r_1,r_2) \in S$
	such that
	\[
		(p_1,p_2) \xrightarrow{u} (q_1,q_2) \text{ and } (p_1,p_2) \xrightarrow{v} (r_1,r_2)
	\]
	for some words $u,v \in \Sigma^*$ with $|u| = |v|$.
	Since for $i \in \{1,2\}$ we have $p_i \xrightarrow{u} q_i$ and $p_i \xrightarrow{v} r_i$, and $\{p_i,r_i,q_i\}$
	is contained in an SCC of $\A_i$ (which is also reachable from the initial state), we have
	\begin{eqnarray*}
            (q_1,q_2) \text{ is final} & \iff & q_1 \text{ and } q_2 \text{ are final} \\
            & \iff & r_1 \text{ and } r_2 \text{ are final} \\
            & \iff & (r_1,r_2) \text{ is final},
         \end{eqnarray*}
	and therefore $\A_1 \times \A_2$ is well-behaved.
\end{proof}
It remains to show the implication from 4.~to 5. 

\begin{lem}
	The class of Boolean combinations of regular left ideals and regular length languages
	is closed under pre-images of $\leftarrow$-transductions.
\end{lem}

\begin{proof}
	For any function $\tau: \Sigma^* \to \Gamma^*$ and $K,L \subseteq \Gamma^*$
	we have $\tau^{-1}(K \cup L) = \tau^{-1}(K) \cup \tau^{-1}(L)$ and $\tau^{-1}(\Gamma^* \setminus L) = \Sigma^* \setminus \tau^{-1}(L)$.
	Now assume that $\tau$ is a $\leftarrow$-transduction. Since it is 
	length-preserving, the $\tau$-pre-image of a length language is again a length language.
	Finally, $\tau$-pre-images of left ideals are left ideals again because $\tau^{-1}(\Gamma^* L) = \Sigma^* \tau^{-1}(L)$.
\end{proof}
It remains to prove that every regular language of polynomial growth  is a Boolean combination
of regular left ideals and regular length languages.
Since a language $L$ and its reversal $L\rev$ have the same growth, we can instead show that every regular language of polynomial growth  is a Boolean combination
of regular right ideals and regular length languages.
The idea is to decompose every regular language of polynomial growth
as a finite union of languages recognized by so called linear cycle automata.

In the following we will allow {\em partial} DFAs $\mathcal{A} = (Q,\Sigma,q_0,\delta,F)$
where $\delta \colon Q \times \Sigma \to Q$ is a partial function.
An SCC $C$ of a partial DFA $\mathcal{A} = (Q,\Sigma,q_0,\delta,F)$
is called a {\em cycle} if for every $p \in C$ there exists at most one $a \in \Sigma$ such that
$\delta(p,a) \in C$.
Note that a singleton SCC $C = \{p\}$ such that $\delta(p,a) \neq p$ whenever 
$\delta(p,a)$ is defined is a cycle, too. Such a cycle is called {\em trivial}.
A partial DFA $\mathcal{A} = (Q,\Sigma,q_0,\delta,F)$ is a {\em linear cycle automaton} if
\begin{itemize}
	\item for all $p,q \in Q$ there exists at most one symbol $a \in \Sigma$ such that $\delta(p,a) = q$,
	\item every SCC $C$ of $\mathcal{A}$ is a (possibly trivial) cycle,
	\item there is an enumeration $C_1, \ldots, C_k$ of the SCCs of $\A$ such that
	      there is exactly one transition from $C_i$ to $C_{i+1}$ for $1 \leq i \leq k-1$, and there
	      is no transition from $C_i$ to $C_j$ for $j > i+1$,
	\item $q_0$ belongs to $C_1$,
	\item $|F| = 1$ and the unique final state belongs to $C_k$. 
\end{itemize}

\begin{lem}
	If $L$ is a regular language with polynomial growth,
	then $L$ is a finite union of languages recognized by linear cycle automata.
\end{lem}

\begin{proof}
	Let $\A = (Q,\Sigma,q_0,\delta,F)$ be the minimal DFA for a regular language $L \subseteq \Sigma^*$ of polynomial growth.
	We first remove from $\A$ all states from which no state in $F$ is reachable; then $\A$ becomes  a partial DFA.
	By \cite[Lemma 2]{GKRS10} for every $q \in Q$ there exists a word $u_q \in \Sigma^*$ such that the language
	$\{ w \in \Sigma^* : \delta(q,w) = q \}$ is a subset of $u_q^*$. Thus,  for every SCC $C$ of $\A$ and every state $q \in C$
	there is at most one symbol $a \in \Sigma$ with $\delta(q,a) \in C$.
		
	A {\em path description} is a sequence
	\[
		P = (p_1,C_1,q_1,a_1, p_2,C_2,q_2,a_2, \dots, p_k,C_k,q_k)
	\]
	where $C_1, \dots, C_k$ is a chain in the partial ordering on the set of SCCs of $\mathcal{A}$,
	$p_1 = q_0$, $p_i,q_i \in C_i$ for all $1 \le i \le k$, $\delta(q_i,a_i) = p_{i+1}$ for all $1 \le i < k$
	and $q_k \in F$.
	Clearly there are only finitely many path descriptions.
	To every accepting run of $\mathcal{A}$ we can assign a path description, which indicates the SCCs
	traversed in the run and the transitions that lead from one SCC to the next SCC.
	We can write $L(\mathcal{A})$ as a finite union	of languages over all path descriptions.
	For every path description $P$, we take the set of all words accepted  by a run of $\A$ whose path
	description is $P$.
						
	Consider a single path description $P= (p_1,C_1,q_1,a_1, p_2,C_2,q_2,a_2, \dots, p_k,C_k,q_k)$ and let 
	$\mathcal{B}$ be the restriction of $\mathcal{A}$ to the SCCs $C_i$.
	Furthermore all transitions between two distinct SCCs are removed except for the transitions
	$(q_i,a_i,p_{i+1})$. 
	Finally, $q_k$ becomes the only final state of $\mathcal{B}$.
	Then $\mathcal{B}$ is indeed a linear cycle automaton.
\end{proof}

\begin{lem}
	Let $\A$ be a linear cycle automaton.
	There are linear cycle automata $\mathcal{A}_1, \dots, \mathcal{A}_s$ such that $L(\mathcal{A}) = \bigcup_{i=1}^s L(\mathcal{A}_i)$
	and in each $\mathcal{A}_i$ each non-trivial cycle has the same length.
\end{lem}

\begin{proof}
	Let $m_1, \dots, m_k$ be the lengths of each non-trivial cycle (SCC) in $\mathcal{A}$ and $m$ be the least common multiple of $m_1, \dots, m_k$.
	The language $L(\mathcal{A})$ is the finite union of all languages accepted by linear cycle automata that are obtained from $\A$ by doing 
	the following replacement for every non-trivial cycle
	\[
		C :   q_1 \xrightarrow{a_1} q_2 \xrightarrow{a_2} q_3 ~ \cdots ~ q_{m_i-1} \xrightarrow{a_{m_i-1}}  q_{m_i} \xrightarrow{a_{m_i}} q_1
	\]
	of $\A$. W.l.o.g. assume that $q_1$ is either the initial state of $\A$ or the target state of the unique transition entering $C$.
	Choose an arbitrary number 
	$0 \le d_i < \frac{m}{m_i}$ (we then take the finite union over all such choices).
	We replace $C$ by a path $P$ of length $d_i m_i$ followed by cycle $C'$ of length $m$, having the form
	\[
		P : q'_1 \xrightarrow{w^{d_i}} q_1,   \quad C' : q_1 \xrightarrow{a_1} q_2 \xrightarrow{a_2} q_3 \cdots q_{m-1} \xrightarrow{a_{m-1}}  q_{m} \xrightarrow {a_m} q_1,
	\]
	where $a_1 a_2 \cdots a_m = (a_1 a_2 \cdots a_{m_i})^{m/m_i}$. All states on the path $P$ except for $q_1$ are new and also all states
	$q_{m_i+1}, \ldots, q_m$ are new. 
	If $q_1$ is the initial state of $\A$ then $q'_1$ is the new initial state. Otherwise,  the unique transition entering $C$ is redirected to the new state $q'_1$.  	
	The union of the languages recognized by all automata of this form is $L(\mathcal{A})$.
\end{proof}

\begin{lem}
	Let $\mathcal{A}$ be a linear cycle automaton in which each non-trivial cycle has the same length.
	Then $L(\mathcal{A})$ is a Boolean combination of regular right ideals and regular length-languages.
\end{lem}

\begin{proof}
	Let $L \subseteq \Sigma^*$ be the language recognized by $\mathcal{A}$.
	There are numbers $p,q \ge 0$ such that each word in $L$ has length $p+qn$ for some $n \ge 0$.
	Here $q$ is the uniform length of the non-trivial cycles in $\mathcal{A}$.
	We claim that $L$ is the intersection of the three languages
	\begin{itemize}
		\item $L \Sigma^*$, which is a regular right ideal,
		\item $\{ x \in \Sigma^* : \mathrm{Pref}(x) \subseteq \mathrm{Pref}(L) \}$, which is the complement of a regular right ideal,
		\item $\Sigma^p  (\Sigma^q)^*$, which is a length language.
	\end{itemize}
	Clearly $L$ is contained in the described intersection. Conversely, consider a word $x$ in the intersection.
	We have $x = yz$ where $y \in L$. Hence, $|y| = p + qn$ for some $n$.
	Since $|x| = p+qn'$ for some $n'$, the length $|z|$ is divided by $q$.
	Since $y \in L$, $\mathcal{A}(y)$ is the unique final state of $\mathcal{A}$, which 
	belongs to the unique maximal SCC $C$ of $\mathcal{A}$.
	If $C$ is non-trivial, then it is a cycle of length $q$ and also $\mathcal{A}(yz)$ is the final state, i.e., $x \in L$.
	If $C$ is trivial, then $y, yz \in L$ implies  $z = \varepsilon$
	and $x$ is also accepted by $\mathcal{A}$.
\end{proof}
This concludes the proof for the direction from 4.~to 5.

\subsection{Lower bounds}

Recall that the space bound in \cref{prop:uniform-algorithm} is exponential in the number $m$ of automaton states.
In the following we show that this bound is tight, already for the fixed-size sliding window model.
For $k \ge 0$ we define the language $L_k \subseteq \{0, \dots, k\}^*$ by
\begin{itemize}
	\item $L_0 = 0^+$, and
	\item $L_k = L_{k-1} \cup L_{k-1} \, k \, \{0, \dots, k-1\}^*$ for $k \geq 1$.
\end{itemize}
Observe that a word $a_1 \cdots a_n \in \{0, \dots, k\}^*$ belongs to $L_k$ if and only if
$n \ge 1$, $a_1 = 0$ and for each $1 \le i \le n$ it holds that $a_i = 0$ or $a_i \neq \max_{1 \le j \le i-1} a_j$.
We can construct a DFA $\A_k$ for $L_k$ with $k+3$ states,
which stores the maximum value seen so far in its state,
see \cref{fig:zimim}.

\tikzstyle{state} = [circle,draw, minimum size = 20pt]
	
\begin{figure}[t]
	\centering
	\begin{tikzpicture}[semithick,->,>=stealth]
		\node[state, initial, initial text={}] (in) {};
		\node[state, right = 2em of in] (0) {0};
		\node[state, right = 2em of 0] (1) {1};
		\node[state, right = 2em of 1] (2) {2};
				
		\node[right = 1em of 2]  {$\cdots$};
				
		\node[state, right = 4em of 2] (s) {$s$};
		\node[state, right = 4em of s] (t) {$t$};
				
		\node[right = 1em of t]  {$\cdots$};
				
		\node[state, right = 4em of t] (k) {$k$};
				
		\draw (in) edge node[above] {\footnotesize $0$} (0);
		\draw [loop below] (0) edge node[below] {\footnotesize $0$} (0);
		\draw (0) edge node[above] {\footnotesize $1$} (1);
		\draw [bend left = 40] (0) edge node[above] {\footnotesize $2$} (2);
		\draw [loop below] (1) edge node[below] {\footnotesize $0$} (1);
		\draw (1) edge node[above] {\footnotesize $2$} (2);
		\draw [loop below] (2) edge node[below] {\footnotesize $0,1$} (2);
												
		\draw [loop below] (s) edge node[below] {\footnotesize $t < s$} (s);
		\draw [bend left = 40] (s) edge node[above] {\footnotesize $t > s$} (t);
		\draw [loop below] (k) edge node[below] {\footnotesize $0, \ldots, k-1$} (k);
	\end{tikzpicture}
	\caption{An automaton for $L_k$. Omitted transitions lead to a sink state. All states are final, except from the sink state.}
	\label{fig:zimim}
\end{figure}
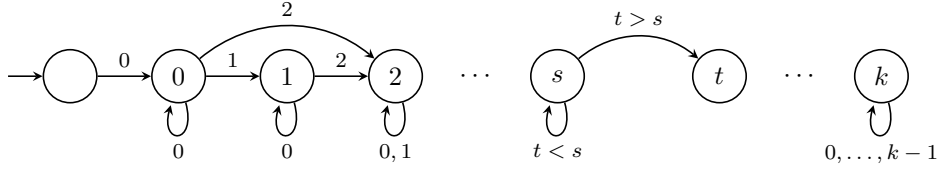

To prove that each $L_k$ belongs to $\V(\mathcal{O}(\log n))$,
we show that $L_k$ is a Boolean combination of regular left ideals.
Given a word $x = a_1 \cdots a_n \in \Sigma^*$ and a language $L \subseteq \Sigma^*$,
a position $1 \le i \le n$ is an {\em $L$-alternation point},
if exactly one of the words $a_i \cdots a_n$ and $a_{i+1} \cdots a_n$ belongs to $L$.
Denote by $\mathrm{alt}_L(x)$ the number of $L$-alternation points in $x$.

\begin{lem}
	\label{lem:alternations}
	Let $L \subseteq \Sigma^*$ be regular.
	Then $L$ is a Boolean combination of at most $k$ regular left ideals
	if and only if $\mathrm{alt}_L(x) \le k$ for all $x \in \Sigma^*$.
\end{lem}

\begin{proof}
	If $L$ is a Boolean combination of regular left ideals $L_1, \dots, L_k$,
	then each $L$-alternation point in a word is an $L_i$-alternation point for some $1 \le i \le k$.
	Since $L_i$ is a left ideal, each word has at most one $L_i$-alternating point and
	we obtain $\mathrm{alt}_L(x) \le k$ for all $x \in \Sigma^*$.
		
	Conversely, assume that $\mathrm{alt}_L(x) \le k$ for all $x \in \Sigma^*$.
	Without loss of generality assume $\varepsilon \in L$, which ensures that $x \in L$ if 
	and only if $\mathrm{alt}_L(x)$ is even. If $\varepsilon \not\in L$, then $x \in L$ if 
	and only if $\mathrm{alt}_L(x)$ is odd, and we can argue similarly as below.
		
	We define $P_i = \{ x \in \Sigma^* : \mathrm{alt}_L(x) \ge i \}$ for $i \ge 0$
	and write $L$ as
	\[
		L = \bigcup_{0 \le i \le k \text{ even}} (P_i \setminus P_{i+1}).
	\] 
	Each $P_i$ is a left ideal because prolonging a word on the left only increases the number of $L$-alternation points.
	Furthermore, each $P_i$ is regular: by  enriching a DFA for $L$ with a counter up to $i$, a DFA can verify that the input
	$x$ satisfies $\mathrm{alt}_L(x) \ge i$.	
	Using the fact that $P_{0} = \Sigma^*$ and $P_i = \emptyset$ for all $i > k$,
	we can write $L$ as
	\[
		L = \begin{cases} (\Sigma^* \setminus P_{1}) \cup (P_{2} \setminus P_{3}) \cup \cdots \cup (P_{k-2} \setminus P_{k-1}) \cup P_{k}, & \text{if $k$ is even}  \\
		(\Sigma^* \setminus P_{1}) \cup (P_{2} \setminus P_{3}) \cup \cdots \cup (P_{k-1} \setminus P_{k}),
		& \text{if $k$ is odd.}
		\end{cases}
	\]
	This proves that $L$ is a Boolean combination of the regular left ideals $P_{1}, \dots, P_{k}$,
	which concludes the proof.
\end{proof}

\begin{lem}
	For all $k \ge 0$ and $x \in \N^*$ we have $\mathrm{alt}_{L_k}(x) \le 2^{k+2} - 2$. Moreover,
	$V_{L_k}(n)  \le (2^{k+3} \cdot (k+3) +1) \cdot \log n  \ + \ c_k$ for $n$ large enough, where $c_k$ only depends on $k$.
\end{lem}
	
\begin{proof}
	We prove the lemma by induction on $k \ge 0$.
	Clearly each word has at most 2 alternation points with respect to $L_0 = 0^+$.
	Now let $k \ge 1$ and $x \in \N^*$.
	If all occurring numbers in $x$ are at most $k-1$,
	then $\mathrm{alt}_{L_k}(x) = \mathrm{alt}_{L_{k-1}}(x)$ and the claim follows by induction.
	Otherwise consider the last occurrence of a number $\ge k$ and factorize $x = y \ell z$
	where $y \in \N^*$, $\ell \ge k$ and $z \in \{0, \ldots, k-1\}^*$.
	If $\ell > k$, then the first $|y|$ positions of $x$ cannot contain $L_k$-alternation points
	and we get 
	\[
		\mathrm{alt}_{L_k}(x) \le 1 + \mathrm{alt}_{L_k}(z) = 1 + \mathrm{alt}_{L_{k-1}}(z) \le 2^{k+1} - 1 \le 2^{k+2} - 2.
	\]
	Now assume $x = y k z$.
	By the definition of $L_k$ each $L_k$-alternation point in $x$ is either (i) an $L_{k-1}$-alternation point in $y$,
	(ii) an $L_{k-1}$-alternation point in $z$, 
	or (iii) position $|y|+1$ (i.e., the last position, where $k$ occurs).
	Hence we have
	\[
		\mathrm{alt}_{L_k}(x)  \le  1 + \mathrm{alt}_{L_{k-1}}(y) + \mathrm{alt}_{L_{k-1}}(z) \le  1 + (2^{k+1} - 2) + (2^{k+1} - 2)  \le 2^{k+2} - 2.                           	
	\]
	From \cref{prop:uniform-algorithm,thm-characterizations,lem:alternations} 
	we obtain $V_{L_k}(n) \le (2^{k+3} \cdot (k+3) +1) \cdot \log n  \ + \ c_k$ for $n$ large enough, where $c_k$ only depends on $k$.
\end{proof}

\begin{thm}
	For each $k \ge 1$ there exists a language $L_k \subseteq \{0, \ldots, k\}^*$
	recognized by a DFA with $k+3$ states
	such that $F_{L_k}(n) \geq (2^k-1) \cdot \log n - c'_k$, where $c'_k$ only depends on $k$.
\end{thm}

\begin{proof}
	Of course, we take the languages $L_k$ considered in this section.
	We define the languages $Z_0 = 0^*$ and $Z_k = Z_{k-1} \; k \; Z_{k-1}$ for $k \geq 1$. An example word from $Z_3$ is $0010002100300010020010$.
	The crucial fact about words $x \in Z_k$ that we are using is the following: Every suffix of $x$ that starts with $0$ belongs to $L_k$ and 
	every suffix of $x$ that starts with $a > 0$ does not belong to $L_k$. The former follows by induction on $k$; the latter holds since words in $L_k$ start with $0$.
	        
	Fix some $k \ge 1$ and
	let $\B=(\B_n)_{n \ge 0}$ be a fixed-size sliding window algorithm for $L_k$ where $S_n$ is the state space of $\B_n$.
	Let us consider window size $n$.
	We claim that $\B_n$ distinguishes all $\binom{n}{2^k-1}$ words in $Z_k$ of length $n$. 
		
	\medskip
	\noindent
	{\em Claim.} Let $x, y \in Z_k$ such that $|x|=|y| = n$ and $x \neq y$. Then $\B_n(x) \neq \B_n(y)$.
		
	\medskip
	\noindent
	In order to get a contradiction,
	consider two words $x,y \in Z_k$ with $|x|=|y| = n$, $x \neq y$, and $\B_n(x) = \B_n(y)$.
	Thus, we can write
	$x = z a u$ and $y = z b v$ with $a,b \in \{0,\ldots,k\}$, $a \neq b$. We must have $a = 0$ and $b > 0$ or vice versa.
	Assume that $a=0$ and $b > 0$. Thus, $au \in L_k$ and $bv \not\in L_k$. Hence, we have
	$\wnd(x0^{|z|}) = au0^{|z|} \in L_k$ and $\wnd(y0^{|z|}) = bv0^{|z|} \not\in L_k$.
	But if $\B_n(x) = \B_n(y)$, then also 
	$\B_n(x0^{|z|}) = \B_n(y0^{|z|})$, which yields a contradiction.
	
	\medskip
	\noindent
	The above claim implies that $\B_n$ has at least $\binom{n}{2^k-1}$ many states.
	Hence, the space complexity of $\B$ is at least
	$$
	\log \binom{n}{2^k-1} \geq \log \bigg( \frac{n}{2^k-1} \bigg)^{2^k-1}  \geq \log \bigg( \frac{n}{2^k} \bigg)^{2^k-1} = (2^k-1) \cdot (\log n - k) .
	$$
	This concludes the proof.
\end{proof}

\section{Constant space algorithms}

\Cref{lem:psi-L} implies that $V_L(n) \geq \log n$ if $\emptyset \neq L \neq \Sigma^*$. 
Thus, only trivial languages have a constant-space variable-size streaming algorithm.
This changes in the fixed-size window model. In \cite{GanardiHL16} we characterized those regular languages
$L$ in $\F(\mathcal{O}(1))$ in terms
of the  left Cayley graph of the syntactic monoid of $L$. Here we give a more natural characterization
that will be used in the next section.

A language $L\subseteq\Sigma^*$ is called {\em $k$-suffix testable}
if for all $x,y \in \Sigma^*$ and $z \in \Sigma^k$ we have
\[
	xz \in L \iff yz \in L.
\]
Equivalently, $L$ is a Boolean combination of languages of the form $\Sigma^* w$ where $w \in \Sigma^{\le k}$.
We call $L$ {\em suffix testable} if it is $k$-suffix testable for some $k \ge 0$. Clearly, every finite language
is suffix testable: if $L \subseteq \Sigma^{\leq k}$ then $L$ is $(k+1)$-suffix testable. 
The class of suffix testable languages corresponds to the variety $\mathbf{D}$ of definite monoids \cite{Str85}.

Recall the languages
\begin{equation} \label{eq-L_n}
	L_n := \{w \in \Sigma^* : \last_n(w) \in L \}
\end{equation}
recognized by a family of streaming algorithms in the fixed-size model.
The main result of this section is:
		
\begin{thm}
	\label{thm:reg-constant}
	A regular language $L \subseteq \Sigma^*$ belongs to $\F(\mathcal{O}(1))$ if and only if
	$L$ is a finite Boolean combination of suffix testable languages and regular length languages.
\end{thm}
The following definitions are useful, which are also studied in~\cite{GawrychowskiJ09}.
For two languages $K,L \subseteq \Sigma^*$, we denote by $K \triangle L=(K\setminus L)\cup(L\setminus K)$
the symmetric difference of $K$ and $L$. We define the distance $d(K,L)$ by
\[
	d(K,L) = \begin{cases}
	\sup \{ |u| : u \in K \triangle L \} + 1, & \text{if } K \neq L, \\
	0, & \text{if } K = L.
	\end{cases}
\]
Notice that $d(K,L)<\infty$ if and only if $K\triangle L$ is finite.
For a DFA $\A = (Q,\Sigma,q_0,\delta,F)$ and a state $p\in Q$, we define $\A_p = (Q,\Sigma,p,\delta,F)$.
For two states $p,q\in Q$, we define the distance $d(p,q)=d(L(\A_p),L(\A_q))$.
It is known that $d(p,q)<\infty$ implies $d(p,q) \le |Q|$, see \cite[Lemma~1]{GawrychowskiJ09}.
	
\begin{lem}
	\label{lem:testable-distance}
	Let $L \subseteq \Sigma^*$ be regular and $\A = (Q,\Sigma,q_0,\delta,F)$ be its minimal DFA. 
	We have:
	\begin{enumerate}[(i)]
		\item $d(p,q) \le k$ if and only if $\delta(p,z) = \delta(q,z)$ for all $p,q \in Q$ and $z \in \Sigma^k$.
		\item $L$ is $k$-suffix testable if and only if $d(p,q) \le k$ for all $p,q \in Q$.
		\item If there exists $k \ge 0$ such that $L$ is $k$-suffix testable, then $L$ is $|Q|$-testable.
	\end{enumerate}
\end{lem}

\begin{proof}
	The proof of (i) is an easy induction: If $k=0$, the statement is $d(p,q) = 0$ iff $p = q$,
	which is true because $\A$ is minimal. For the induction step, we have
	$d(p,q) \le k+1$ iff $d(\delta(p,a),\delta(q,a)) \le k$ for all $a \in \Sigma$ iff $\delta(p,z) = \delta(q,z)$
	for all $z \in \Sigma^{k+1}$.
					
	For (ii), assume that $L$ is $k$-suffix testable and consider two states $p = \A(x)$ and $q= \A(y)$.
	If $z \in L(\A_p) \triangle L(\A_q)$, then $|z| < k$ because $xz \in L$ iff $yz \notin L$
	and $L$ is $k$-suffix testable.
						
	Now assume that $d(p,q) \le k$ for all $p,q \in Q$ and consider $x,y \in \Sigma^*$, $z \in \Sigma^k$.
	Since $d(\A(x),\A(y)) \le k$, (i) implies $\A(xz) = \A(yz)$, and in particular $xz \in L$ iff $yz \in L$.
	Therefore, $L$ is $k$-suffix testable.

	Point (iii) follows from (ii) and the above mentioned results from \cite[Lemma~1]{GawrychowskiJ09}.
\end{proof}

\begin{thm}\label{thm:suffix-language}
	For any $L \subseteq \Sigma^*$ and $n \ge 0$,
	the language $L_n$ from \eqref{eq-L_n} is $(2^{F_L(n)+1}-1)$-suffix testable.
\end{thm}
\begin{proof}
	Let $(\A_n)_{n \ge 0}$ be an optimal fixed-size sliding window algorithm for $L$
	where $\A_n=(S_n,\Sigma,s_n,\delta_n,F_n)$, which is the minimal DFA for $L_n$.
	For every $n$, the language $L_n$ is $n$-suffix testable because
	\[
		L_n = \{ w \in \Sigma^{\le n-1} : \last_n(w) \in L \} \cup \Sigma^* (L \cap \Sigma^n).
	\]
	By \cref{lem:testable-distance}(iii) every language $L_n$ is $|S_n|$-suffix testable.
	Together with $F_L(n) = \log |S_n|$ this proves the claim for $L_n$.
\end{proof}

\begin{cor}\label{cor:constant-suffix}
	A language $L \subseteq \Sigma^*$ belongs to $\F(\mathcal{O}(1))$ if and only if
	there exists a $k \ge 0$ such that $L_n$ is $k$-suffix testable for all $n \ge 0$.
\end{cor}

\begin{proof}
	The left-to-right direction follows from \cref{thm:suffix-language}.
	If each $L_n$ is $k$-suffix testable, then the streaming algorithm for window length $n$
	only needs to maintain the last $k$ symbols to test membership of a word of length $n$ in $L_n$, or equivalently in $L$.
\end{proof}

\begin{proof}[Proof of \cref{thm:reg-constant}]
	First, let $L \subseteq \Sigma^*$ be a regular language in $\F(\mathcal{O}(1))$.
	By \cref{thm:suffix-language} there exists $k \ge 0$ such that
	$L_n$ is $k$-suffix testable for all $n \ge 0$.
	We write $L$ as the Boolean combination
	\[
		L = (L \cap \Sigma^{\le k-1}) \cup \bigcup_{z \in \Sigma^k} (Lz^{-1}) \, z = (L \cap \Sigma^{\le k-1}) \cup \bigcup_{z \in \Sigma^k} ((Lz^{-1}) \, \Sigma^k \cap \Sigma^* z)
	\]
	where $Lz^{-1} = \{ x \in \Sigma^* : xz \in L \}$ is the regular right quotient of $L$ by $z$.
	The set $L \cap \Sigma^{\le k-1}$ is finite and hence suffix testable. 
	It remains to show that each $Lz^{-1}$ is a length language.
	Consider two words $x,y \in \Sigma^*$ of the same length $|x|=|y|=n$.
	Since $|xz|=|yz| = n+k$ and $L_{n+k}$ is $k$-suffix testable we have $xz \in L$ iff $yz \in L$,
	and hence $x \in Lz^{-1}$ iff $y \in Lz^{-1}$.
			
	For the other direction note that:
	\begin{itemize}
		\item if $L$  is a length language or suffix testable language then clearly $L \in \F(\mathcal{O}(1))$, and
		\item $\F(\mathcal{O}(1))$ is closed under Boolean operations by \cref{lem:boolean}.
	\end{itemize}
	This proves the theorem.
\end{proof}
We give another characterization of the regular languages in $\F(\mathcal{O}(1))$, which yields a decision procedure in the next section.

\begin{prop}
	\label{prop:constant-decision}
	Let $L \subseteq \Sigma^*$ be regular and $\A = (Q,\Sigma,q_0,\delta,F)$ be its minimal DFA.
	Then $L \in \F(\mathcal{O}(1))$ if and only if for all $x,y \in \Sigma^*$ with $|x|=|y|$ and $z \in \Sigma^{|Q|}$ we have $\A(xz) = \A(yz)$.
\end{prop}

\begin{proof}
	Assume that $L \in \F(\mathcal{O}(1))$. By \cref{cor:constant-suffix} there exists $k \ge 0$ such that each $L_n$ is $k$-suffix testable.
	Let $x,y \in \Sigma^*$ with $|x|=|y|=n$. For all $z \in \Sigma^k$ we have $xz \in L_{n+k}$ iff $yz \in L_{n+k}$.
	Thus, $xz \in L$ iff $yz \in L$ and hence $d(\A(x),\A(y)) \le k$.
	This implies $d(\A(x),\A(y)) \le |Q|$. By \cref{lem:testable-distance}(i)
	we have $\A(xz) = \A(yz)$ for all $z \in \Sigma^{|Q|}$.
						
	Conversely, assume that $\A(xz) = \A(yz)$ for all $x,y \in \Sigma^*$ with $|x|=|y|$ and $z \in \Sigma^{|Q|}$.
	This means that one can simulate the automaton on the active window by only storing the last $|Q|$ many symbols and hence in space $\mathcal{O}(1)$.
\end{proof}

\section{Deciding space complexity in the sliding window model} \label{sec-computing-space}

\renewcommand{\star}{^{*}}

In this section, we consider the complexity of the following  decision problems:
\begin{itemize}
	\item
	      \textsc{Dfa}$(1)$:
	      Given a DFA $\A$, does $L(\A) \in \F(\mathcal{O}(1))$ hold?
	\item
	      \textsc{Nfa}$(1)$:
	      Given an NFA $\A$, does $L(\A)\in\F(\mathcal{O}(1))$ hold?
	\item
	      \textsc{Dfa}$(\log n)$:
	      Given a DFA $\A$, does $L(\A) \in \F(\mathcal{O}(\log n)) = \V(\mathcal{O}(\log n))$ hold?
	\item
	      \textsc{Nfa}$(\log n)$:
	      Given an NFA $\A$, does $L(\A) \in \F(\mathcal{O}(\log n)) = \V(\mathcal{O}(\log n))$ hold?
\end{itemize}
It is straightforward to show that membership in $\V(\mathcal{O}(1))$ for a regular language
that is given by a 
DFA (resp., an NFA) is $\NL$-complete (resp., $\PSPACE$-complete):
By \cref{thm:trichotomy} one has to check whether $L(\A) = \emptyset$ or 
$L(\A) = \Sigma\star$, and universality for DFAs (resp., NFAs) is 
$\NL$-complete (resp., $\PSPACE$-complete).
In \cref{sec-dfa-case} (resp., \cref{sec-nfa-case})  we show that \textsc{Dfa}$(1)$ and \textsc{Dfa}$(\log n)$
(resp., \textsc{Nfa}$(1)$ and \textsc{Nfa}$(\log n)$) are \NL-complete (resp., \PSPACE-complete).

\subsection{The DFA case} \label{sec-dfa-case}

We start with the \NL-hardness for the DFA case:

\begin{thm} \label{thm-NL-hard}
	\textsc{Dfa}$(1)$ and \textsc{Dfa}$(\log n)$ are \NL-hard.
\end{thm}

\begin{proof}
	We reduce from the \NL-complete reachability problem in finite directed graphs.
	Given a finite directed graph $G = (V,E)$ and two vertices $s, t \in V$,
	the question is whether there exists a path from $s$ to $t$.
	We can assume that $s \neq t$ and that each vertex $v \in V$ has exactly two successors $v_a, v_b \in V$.
	Let $\A = (V \cup \{\bot\},\{a,b,c\},s,\delta,\{t\})$ be a DFA
	where
	\[
		\delta(v,x) = \begin{cases}
		v_x & \text{if } v \in V \setminus \{t\}, \, x \in \{a,b\}, \\
		t & \text{if } v = t, \, x \in \{a,b,c\}, \\
		\bot & \text{otherwise.}
		\end{cases}
	\]
	Since $s \neq t$, we can write $L(\A)$ as  $K \, \{a,b,c\}^*$ for some $K \subseteq \{a,b\}^+$.
	Furthermore, there exists a path from $s$ to $t$ in $G$ if and only if $K \neq \emptyset$.
	If $K = \emptyset$, then $L(\A) = \emptyset$ belongs to $\F(\mathcal{O}(1))$ and to $\V(\mathcal{O}(\log n))$.
	If $K \neq \emptyset$, then we claim that $L(\A) = K\; \{a,b,c\}^*$ does not belong to $\V(\mathcal{O}(\log n))$.
	Consider a variable-size sliding window algorithm $\mathcal{M}$ for $L(\A)$.
	Fix an arbitrary word $x \in K$ and let $k=|x|$. Moreover, let $n \ge 0$ and consider the set $\{x,c^k\}^n$ of size $2^n$.
	Let us read two distinct words from $\{x,c^k\}^n$ into two instances of $\mathcal{M}$. We can write these
	words as $uxw$ and $vc^kw$ for some $u,v,w \in \{a,b,c\}^*$ with $|u| = |v|$.
	By removing the first $|u|=|v|$ many symbols from the window, we obtain the active windows
	$xw \in L(\A)$ and $c^kw \notin L(\A)$ and therefore $\mathcal{M}(uxw) \neq \mathcal{M}(vc^kw)$.
	Hence, $\mathcal{M}$ must contain at least $2^n$ many states that are reachable by words of length $k n$. 
	This implies that $L(\A) \notin \V(o(n))$.
\end{proof}

\begin{thm} \label{thm-NL-O(1)}
	\textsc{Dfa}$(1)$ is \NL-complete.
\end{thm}
	
\begin{proof}
	Let us first assume that the input DFA $\A$ is  minimal. Later, we will argue how to handle the general case.
	Since nondeterministic logspace is closed under complement,
	it suffices to decide whether $L(\A) \notin \F(\mathcal{O}(1))$.
	By \cref{prop:constant-decision} this is the case if and only if
	there exist words $x,y,z \in \Sigma^*$ such that $|x|=|y|$, $|z| = |Q|$
	and $\A(xz) \neq \A(yz)$. The existence of such words can be easily verified in nondeterministic
	logspace: One simulates $\A$ on two words of the same length (the words $x,y$), and thereby only stores
	the current state pair. At every time instant, the algorithm can nondeterministically
	decide to continue the simulation from the current state pair $(p,q)$ with a single word (the word $z$)
	for $|Q|$ steps. The algorithm accepts if at the end the two states are distinct. 
		
	The general case, where $\A$ is not minimal is handled as follows:
	Assume that $\A = (\{1, \ldots, k\}, \Sigma,1,\delta,F)$ is the input DFA.
	It is known that DFA equivalence is in \NL  \cite{CH92}. Hence, one can test
	in nondeterministic logspace, whether two states $p,q \in Q$ are equivalent 
	(in the sense that $\delta(p,w) \in F$ iff $\delta(q,w) \in F$ for all $w \in \Sigma^*$).
	We will use this problem as an \NL-oracle in the above \NL-algorithm for minimal DFAs.
	More precisely, let 
	$\A' = (Q,\Sigma,1,\delta',F')$ be the minimal DFA for $\A$, where we assume that $Q$ is the set 
	of all states $q \in \{1, \ldots, k\}$ such that there is no state $p<q$ that is equivalent to $q$.
	We run the \NL-algorithm above for minimal DFAs on $\A'$ without explicitly constructing
	$\A'$. If we have to compute a successor state $\delta'(q,a)$ (where $q \in Q$) we compute, using
	the above \NL-oracle the smallest state that is equivalent to $\delta(q,a)$.
		
	The above argument shows that \textsc{Dfa}$(1)$ belongs to 
	$\NL^{\NL}$. Finally, we use the well-known identity $\NL = \NL^{\NL}$ \cite{Immerman88}.
\end{proof}
In the rest of the section, we show that one can also decide in
nondeterministic logspace whether $L \in \V(\mathcal{O}(\log n))$ (or equivalently $L \in \F(\mathcal{O}(\log n))$).  
As in the proof of \cref{thm-NL-O(1)} we can assume that $L$ is given by its minimal DFA $\A$.

For words $u,x_0,x_1 \in \Sigma^*$ we define 
\[
	Q(u,x_0,x_1) = \{ \A(u x) : x \in \{ x_0,x_1 \}^* \},
\]
which is the set of states of $\A$ reachable from the initial state by first reading $u$ and then an arbitrary product of copies of $x_0$ and $x_1$.
	
\begin{lem}
	\label{lemma:disjoint}
	We have $V_L(n) \in \Theta(n)$ if and only if there are words
	$u_0,u_1, v_0,v_1 \in \Sigma^*$ such that $|u_0|=|u_1| \geq 1$ and 
	$Q(u_0,v_0u_0,v_1u_1) \cap Q(u_1,v_0u_0,v_1u_1) = \emptyset$.
\end{lem}
	
\begin{proof}
	Let $\B$ be the minimal DFA for $L\rev$. If $V_L  \in \Theta(n)$, then $\B$ is not well-behaved,
	i.e., there are words $u, u_0, u_1, v_0, v_1 \in \Sigma^*$ 
	such that
	\begin{itemize}
		\item $|u_0| = |u_1|$,
		\item $\B(uu_0v_0) = \B(u) = \B(uu_1v_1)$,
		\item $\B(uu_0) \notin F$ and $\B(uu_1) \in F$ (and thus $u_0, u_1 \in \Sigma^+$).
	\end{itemize}
	Setting $K = \{u_0v_0, u_1v_1\}^*$ we get $u K u_0 \cap L\rev = \emptyset$ and $u K u_1 \subseteq L\rev$.
	Hence for all $w_0 \in u_0\rev K\rev$ and $w_1 \in u_1\rev K\rev$ we have $w_0 \not\sim_L w_1$.
	Since $\mathcal{A}$ is minimal, this implies 
	$\{ \A(w) : w\in u_0\rev K\rev \} \cap \{ \A(w) : w \in u_1\rev K\rev \} = \emptyset$, i.e.,
	$Q(u_0\rev,v_0\rev u_0\rev,v_1\rev u_1\rev) \cap Q(u_1\rev,v_0\rev u_0\rev,v_1\rev u_1\rev) = \emptyset$.
	
	Next, assume that $Q(u_0,v_0u_0,v_1u_1) \cap Q(u_1,v_0u_0,v_1u_1) = \emptyset$ 
	and $|u_0|=|u_1| \geq 1$ for words $u_0,u_1,v_0,v_1$.
	We clearly have $u_0 \neq u_1$ and hence $v_0 u_0 \neq v_1 u_1$.
	Further, we can choose numbers $p, q \ge 1$ such that $(v_0u_0)^p$ and $(v_1u_1)^q$ have the same length.
	We redefine $v_0$ to be $(v_0u_0)^{p-1} v_0$ and $v_1$ to be $(v_1u_1)^{q-1} v_1$. Thus, $|v_0u_0| = |v_1 u_1|$.
	Moreover, the new resulting sets $Q(u_i,v_0u_0,v_1u_1)$ are contained in the original sets, and are therefore also disjoint.
	Let $c = |v_0u_0| = |v_1 u_1| \geq 1$.
							
	Now consider a variable-size sliding window algorithm $\mathcal{M}$ for $L$ and let $n$ be arbitrary.
	We claim that for all $w_0, w_1 \in \{v_0u_0,v_1u_1\}^n$ with $w_0 \neq w_1$, we have
	$\mathcal{M}(w_0) \neq \mathcal{M}(w_1)$. 
	This is because after removing a suitable number of symbols,
	the active windows contain words $x_0 \in u_0 \{v_0u_0,v_1u_1\}^*$ and $x_1 \in u_1 \{v_0u_0,v_1u_1\}^*$, respectively.
	By assumption, reading $x_0$ and $x_1$ in the minimal DFA $\mathcal{A}$ leads to different states.
	Hence there exists a word $z \in \Sigma^*$ such that $x_0 z \in L$ if and only if $x_1 z \notin L$. Thus, we must
	have $\mathcal{M}(x_0) \neq \mathcal{M}(x_1)$ and therefore $\mathcal{M}(w_0) \neq \mathcal{M}(w_1)$.
		
	Since $\{v_0u_0,v_1u_1\}^n$ consists of $2^n$ many words,
	there exists $w \in \{v_0u_0,v_1u_1\}^n$ such that the encoding of $\mathcal{M}(w)$ has
	length at least $n$. Since $|w| = c n$, we have $V_L(n) \geq \lfloor n/c \rfloor$.	
\end{proof}
We call a tuple $(u_0,u_1,w_0,w_1)$ of words {\em critical}, if $|u_0|=|u_1| \geq 1$,
$u_i$ is a suffix of $w_i$ for all $i \in \{0,1\}$ and $Q(u_0,w_0,w_1) \cap Q(u_1,w_0,w_1) = \emptyset$.
Clearly, the condition from \cref{lemma:disjoint} is equivalent to the existence of a critical tuple.

\begin{lem}
	If there exists a critical tuple,
	then there exists a critical tuple $(u_0,u_1,w_0,w_1)$ such that
	$Q(u_0,w_0,w_1)$ and $Q(u_1,w_0,w_1)$ have each size at most three.
\end{lem}
	
\begin{proof}
	Let $h \colon \Sigma^* \to M$ be the canonical homomorphism into the transition monoid $M$ of $\mathcal{A}$,
	which right acts on $Q$ via $Q \times M \to Q$, $(q,m) \mapsto q \cdot m = m(q)$.
	Assume that $(u_0,u_1,w_0,w_1)$ is a critical tuple.
	Notice that
	\[
		Q(u_i,w_0,w_1) = \{ \A(u_i) \cdot m : m \in \{h(w_0),h(w_1)\}^* \}
	\]
	where $X^*$ denotes the submonoid of $M$ generated by a set $X \subseteq M$.
	It suffices to define a new critical tuple $(u_0,u_1,x_0,x_1)$ with the property that
	$h(x_i) \cdot h(x_j) = h(x_j)$ for all $i,j \in \{0,1\}$.
	This implies $\{h(x_0),h(x_1)\}^* = \{1, h(x_0), h(x_1)\}$,
	and hence, $Q(u_i,x_0,x_1)$ contains at most three elements for both $i \in \{0,1\}$.
							
	Notice that if $(u_0,u_1,w_0,w_1)$ is critical, then also $(u_0,u_1,y_0w_0,y_1w_1)$ is critical
	for all $y_0,y_1 \in \{w_0,w_1\}^*$.
	Let $\omega \ge 1$ be a number such that $m^\omega$ is idempotent for all $m \in M$.
	By choosing $e_0 = (h(w_0)^\omega h(w_1)^\omega)^\omega h(w_0)^\omega$
	and $e_1 = (h(w_0)^\omega h(w_1)^\omega)^\omega$ we indeed obtain $e_i e_j = e_j$ for all $i,j \in \{0,1\}$.	
	Hence we define $x_0 = (w_0^\omega w_1^\omega)^\omega w_0^\omega$ and $x_1 = (w_0^\omega w_1^\omega)^\omega$.
\end{proof}

\tikzstyle{state} = [circle,draw,inner sep = 2pt]
	
\begin{figure}
	\begin{center}
		\begin{tikzpicture}[semithick,->,>=stealth]
			\node[state, initial, initial text={}] (a) {};
			\node[state, above right = 3em and 4em of a] (b) {\footnotesize $p$};
			\node[state, below right = 3em and 4em of a] (c) {\footnotesize $r$};
			\node[state, above right = 1.5em and 5em of b] (b0) {\footnotesize $p_0$};
			\node[state, below right = 1.5em and 5em of b] (b1) {\footnotesize $p_1$};
			\node[state, above right = 1.5em and 5em of c] (c0) {\footnotesize $r_0$};
			\node[state, below right = 1.5em and 5em of c] (c1) {\footnotesize $r_1$};
																	
			\draw (a) edge node[above left] {\footnotesize $u_0$} (b);
			\draw (a) edge node[below left] {\footnotesize $u_1$} (c);
																	
			\draw (b) edge node[above left] {\footnotesize $w_0$} (b0);
			\draw (b) edge node[below left] {\footnotesize $w_1$} (b1);
			\draw[loop right] (b0) edge node[right] {\footnotesize $w_0$} (b0);
			\draw[loop right] (b1) edge node[right] {\footnotesize $w_1$} (b1);
			\draw[bend left = 15] (b0) edge node[right] {\footnotesize $w_1$} (b1);
			\draw[bend left = 15] (b1) edge node[left] {\footnotesize $w_0$} (b0);
															
			\draw (c) edge node[above left] {\footnotesize $w_0$} (c0);
			\draw (c) edge node[below left] {\footnotesize $w_1$} (c1);
			\draw[loop right] (c0) edge node[right] {\footnotesize $w_0$} (c0);
			\draw[loop right] (c1) edge node[right] {\footnotesize $w_1$} (c1);
			\draw[bend left = 15] (c0) edge node[right] {\footnotesize $w_1$} (c1);
			\draw[bend left = 15] (c1) edge node[left] {\footnotesize $w_0$} (c0);
		\end{tikzpicture}
	\end{center}
	\caption{A critical tuple $(u_0,u_1,w_0,w_1)$.}
	\label{fig:critical}
\end{figure}
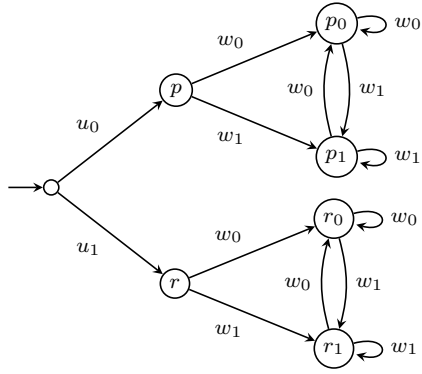
	
\begin{lem} \label{lemma-crit-tripel-NL}
	Given a minimal DFA $\mathcal{A}$, one can test in nondeterministic logspace
	whether $\mathcal{A}$ has a critical tuple.
\end{lem}
	
\begin{proof}
	Let  $\A = (Q,\Sigma,q_0,\delta,F)$ be a minimal DFA.
	\cref{fig:critical} illustrates the structure we need to detect in $\A$.
	To do so, we reduce to testing emptiness of one-counter automata,
	which is known to be decidable in nondeterministic logspace \cite{Latteux}.
	For two states $p,r \in Q$ let $\A_{p,r} = (Q,\Sigma,p,\delta,\{r\})$,
	i.e., the automaton $\mathcal{A}$ with initial state $p$ and final state $r$, and let
	$L(p,r) = L(\A_{p,r})$.
							
	The algorithm iterates over all disjoint sets $\{p,p_0,p_1\}, \{r,r_0,r_1\} \subseteq Q$.
	For $i \in \{0,1\}$ let $\mathcal{A}_i$ be a DFA for the language
	\[
		L(p,p_i) \cap L(p_0,p_i) \cap L(p_1,p_i) \cap L(r,r_i) \cap L(r_0,r_i) \cap L(r_1,r_i).
	\]
	Now consider the language
	\begin{align*}
		\{ v_0 \, \# \, u_0 \, \# \, v_1 \, \# \, u_1 \colon &  v_iu_i \in L(\A_i) \text{ for $i \in \{0,1\}$},  |u_0|=|u_1| \geq 1, \\
		& u_0 \in L(q_0,p), u_1 \in L(q_0,r)  \}           
	\end{align*}
	for which one can construct in logspace a one-counter automaton. The counter is used to verify the constraint 
	$|u_0|=|u_1|$.
	The language above is empty if and only if $\mathcal{A}$ has a critical tuple.
\end{proof}
\cref{lemma:disjoint}, \cref{lemma-crit-tripel-NL} and \cref{thm:trichotomy} together imply:
	
\begin{cor} \label{cor-dfa-lo}
	\textsc{Dfa}$(\log n)$ is $\NL$-complete.
\end{cor}

\subsection{The NFA case} \label{sec-nfa-case}

\newcommand{\aut}[1]{\mathcal{#1}}

In this section, we show that the problems \textsc{Nfa}$(1)$ and \textsc{Nfa}$(\log n)$ are both \PSPACE-complete.
The upper bounds follow easily from \cref{thm-NL-O(1)} and \cref{cor-dfa-lo} and the following fact (see \cite[Lemma~1]{LoMa11regular}): If
a mapping $f$ can be computed by a Turing-machine with a polynomially bounded work tape (the output can be of exponential
size) and $L$ is a language that can be decided in polylogarithmic space, then $f^{-1}(L)$ belongs to \PSPACE. Note that from
a given NFA $\A$ one can compute an equivalent DFA using polynomially bounded work space: One iterates over all subsets 
of the state set of $\A$; the current subset is stored on the work tape. For every subset and input symbol one then writes
the corresponding transition of the DFA on the output tape.

\begin{thm}
	\label{thm:constFnfa}
	\textsc{Nfa}$(1)$ is $\PSPACE$-complete.
\end{thm}
\begin{proof}
	By the above remark it suffices to
	establish $\PSPACE$-hardness of \textsc{Nfa}$(1)$. For this we will reduce the NFA universality problem to \textsc{Nfa}$(1)$.
	The NFA universality problem is $\PSPACE$-complete \cite{MS72}. 
	W.l.o.g.\ consider the alphabet $\Sigma = \{a,b\}$. For an NFA $\A = (Q,\Sigma,I,\Delta,F)$ we define $\rho(\A)$ to be the automaton that results from $\A$ by adding a new initial state $\bar q$ with an $a$-labeled self-loop and a $b$-labeled transition from every state of $F$ to $\bar q$. The only final state of $\rho(\A)$ is $\bar q$. More formally, we define $\rho(\A)$ as follows:
	\[
		\rho(\A) =
		(Q \cup \{\bar q\},
		\Sigma,
		I \cup \{\bar q\},
		\Delta \cup \{(q,b,\bar q) \mid q \in F \} \cup \{(\bar q,a,\bar q)\},
		\{\bar q\}
		).
	\]
	Notice that the $\rho$-construction implies $L(\rho(\A)) = a\star \cup L(\A)  \, b\,  a\star$. It is then easy to verify that $L(\A) = \Sigma\star$ iff $L(\rho(\A)) = \Sigma\star$.
	If $L(\A) = \Sigma^*$ then clearly $L(\rho(\A)) = \Sigma^* \in F(1)$.
	Conversely, assume that $L(\rho(\A)) = a\star \cup L(\A)  \, b\,  a\star$ belongs to $F(1)$.
	By \cref{thm:reg-constant} there exists a number $k \in \N$ such that $a\star \cup L(\A)  \, b\,  a\star$
	is a Boolean combination of $k$-suffix testable languages and regular length languages.
	Let $x \in \{a,b\}^n$ be any word of length $n$.
	Since $a^{n+1+k} \in L(\rho(\A))$ and $xba^k$ share the same $k$-suffix and are of the same length, we also know that $xba^k \in L(\rho(\A))$
	and hence $x \in L(\A)$. This proves that $\A$ is universal.

	We have thus established that the polynomial-time (in fact, log-space) construction $\A \mapsto \rho(\A)$ reduces the universality problem for NFAs to \textsc{Nfa}$(1)$.
\end{proof}

\begin{thm}
	\label{thm:logFVnfa}
	\textsc{Nfa}$(\log n)$ is $\PSPACE$-complete.
\end{thm}
\begin{proof}
	It remains to show that \textsc{Nfa}$(\log n)$ is \PSPACE-hard, which can be shown by reducing the NFA universality problem to \textsc{Nfa}$(\log n)$. W.l.o.g.\ the alphabet of the input automaton is $\Sigma = \{a,b\}$, and we also consider the extended alphabet $\Gamma = \{a,b,c\}$. For an NFA $\A = (Q,\Sigma,I,\Delta,F)$ we define $\rho(\A)$ to be the automaton that results from $\A$ by adding a new initial and final state $\bar q$ with $a$- and $b$-labeled self-loops, a $c$-labeled transition from every state of $F$ to $\bar q$, and a $c$-labeled transition from $\bar q$ to every state of $\A$. The only final state of $\rho(\A)$ is $\bar q$.
	More formally, we define
	\begin{align*}
		\rho(\A) = {}     &   
		(Q \cup \{\bar q\},
		\Gamma,
		I \cup \{\bar q\},
		\rho(\Delta),
		\{\bar q\}
		), \text{ where}
		\\
		\rho(\Delta) = {} &   
		\Delta \cup
		\{(q,c,\bar q) \mid q \in F \} \cup
		\{(\bar q,c,q) \mid q \in Q \} \cup
		\{(\bar q,x,\bar q) \mid x \in \{a,b\} \} .
	\end{align*}
	The automaton $\sigma(\A)$ results from $\A$ be adding a new initial and final state $\bar q$ with $a$- and $b$-labeled self-loops, a $c$-labeled transition from $\bar q$ to each initial state of $\A$, and a $c$-labeled transition from every state of $\A$ to $\bar q$. The only initial state of $\sigma(\A)$ is $\bar q$. More formally, we define $\sigma(\A)$ as follows:
	\begin{align*}
		\sigma(\A) = {}     &   
		(Q \cup \{\bar q\},
		\Gamma,
		\{\bar q\},
		\sigma(\Delta),
		F \cup \{\bar q\}
		), \text{ where}
		\\
		\sigma(\Delta) = {} &   
		\Delta \cup
		\{(\bar q,c,q) \mid q \in I \} \cup
		\{(q,c,\bar q) \mid q \in Q \} \cup
		\{(\bar q,x,\bar q) \mid x \in \{a,b\} \} .
	\end{align*}
	Then, we have $\rho(\A)\rev = \sigma(\A\rev)$, which is also equal to
	\begin{align*}
		\rho(\A)\rev = {}     &   
		(Q \cup \{\bar q\},
		\Gamma,
		\{\bar q\},
		\rho(\Delta)\rev,
		I \cup \{\bar q\}
		) \text{ with}
		\\
		\rho(\Delta)\rev = {} &   
		\Delta\rev \cup
		\{(\bar q,c,q) \mid q \in F \} \cup
		\{(q,c,\bar q) \mid q \in Q \} \cup
		\{(\bar q,x,\bar q) \mid x \in \{a,b\} \} .
	\end{align*}
	Notice that for a deterministic $\A$ (over the alphabet $\Sigma$), the automaton $\sigma(\A)$ (over the alphabet $\Gamma$) is also deterministic. For a nondeterministic automaton $\A$ that satisfies additionally the condition that the state $\emptyset$ is not reachable from the initial state of $\A\det$
	(that is, $\A\det$ does not have the state $\emptyset$), we have that $\sigma(\A)\det \cong \sigma(\A\det)$ (identifying $\bar q$ with $\{\bar q\}$). So,
	\[
		\sigma(\A\revdet) \cong
		\sigma(\A\rev)\det =
		\rho(\A)\revdet
	\]
	under the previously mentioned condition for $\A\rev$, which can be satisfied w.l.o.g.\ by adding an initial non-final state to $\A\rev$ with $a$- and $b$-labeled self-loops (that is, by adding a final non-initial state to $\A$ with $a$- and $b$-labeled self-loops). So, for a nondeterministic automaton $\A$ the claim is:
	\begin{allowdisplaybreaks}
		\begin{eqnarray*}
			L(\A) = \Sigma\star \ & \Leftrightarrow & \  L(\A\revdet) = \Sigma\star  \\
			\ & \Leftrightarrow & \ L(\sigma(\A\revdet)) = \Gamma\star  \\
			\ & \Leftrightarrow & \ \text{the DFA $\sigma(\A\revdet) \cong   \rho(\A)\revdet$ is well-behaved} \\
			\ & \Leftrightarrow & \  \text{$L(\rho(\A)) \in \V(\mathcal{O}(\log n))$.}
		\end{eqnarray*}
	\end{allowdisplaybreaks}%
	The proof of the third equivalence uses the fact that $\sigma(\A\revdet)$ consists of a single SCC. The left-to-right direction is immediate. For the right-to-left direction, observe that the initial and final state $\bar q$ of the DFA $\sigma(\A\revdet)$ has $a$- and $b$-labeled self-loops, i.e., $\Sigma^* \subseteq L(\sigma(\A\revdet))$.
	Thus, for every $n$ there is a word of length $n$ that is accepted from the initial state $\bar q$. But $\sigma(\A\revdet)$ is well-behaved, which implies that all strings of any length must be accepted from $\bar q$, i.e., $L(\sigma(\A\revdet)) = \Gamma^*$.
	So, we have established that the polynomial-time (in fact, log-space) construction $\A \mapsto \rho(\A)$ reduces the universality problem for NFAs to \textsc{Nfa}$(\log n)$.
\end{proof}

\bibliographystyle{plainurl}
\bibliography{bib}

\end{document}